\documentclass[sigconf]{acmart}

\usepackage{multirow}
\usepackage{tabularx}
\usepackage{booktabs}
\usepackage{makecell}
\usepackage{algorithmic}
\usepackage{algorithm}
\usepackage{cancel}

\usepackage{mathtools}

\newcommand{\cut}{\textbf{cut}}
\newcommand{\vol}{\textbf{vol}}

\providecommand{\mat}[1]{\boldsymbol{\mathrm{#1}}}%
\renewcommand{\vec}[1]{\boldsymbol{\mathrm{#1}}}

\DeclareMathOperator*{\minimize}{minimize}

\providecommand{\subjectto}{\ensuremath{\text{subject to}}}

\providecommand{\mA}{\ensuremath{\mat{A}}}
\providecommand{\mB}{\ensuremath{\mat{B}}}

\providecommand{\mD}{\ensuremath{\mat{D}}}

\providecommand{\mI}{\ensuremath{\mat{I}}}

\providecommand{\vd}{\ensuremath{\vec{d}}}

\providecommand{\vg}{\ensuremath{\vec{g}}}

\providecommand{\vr}{\ensuremath{\vec{r}}}

\providecommand{\vw}{\ensuremath{\vec{w}}}
\providecommand{\vx}{\ensuremath{\vec{x}}}

\providecommand{\eps}{\varepsilon}

\newcommand{\confspace}{\vspace}

\newcommand{\mypara}{\textbf}
\newcommand{\Meng}[1]{{\textcolor{blue}{[Meng: #1]}}}

\newcommand{\pan}[1]{{\textcolor{red}{P: #1}}}
\newcommand{\hidecomments}[1]{}
\newcommand{\MINone}[3]{\begin{array}{ll} \displaystyle \minimize_{#1} & {#2} \\ \subjectto & {#3} \end{array}}

\newenvironment{myindentpar}[1]%
  {\begin{list}{}%
          {\setlength{\leftmargin}{#1}}%
          \item[]%
  }
  {\end{list}
}
\AtBeginDocument{%
  \providecommand\BibTeX{{%
    \normalfont B\kern-0.5em{\scshape i\kern-0.25em b}\kern-0.8em\TeX}}}

\AtBeginDocument{%
  \providecommand\BibTeX{{%
    \normalfont B\kern-0.5em{\scshape i\kern-0.25em b}\kern-0.8em\TeX}}}

\copyrightyear{2021}
\acmYear{2021} 
\setcopyright{iw3c2w3}
\acmConference[WWW '21]{Proceedings of the Web Conference 2021}{April 19--23, 2021}{Ljubljana, Slovenia} 
\acmBooktitle{Proceedings of the Web Conference 2021 (WWW '21), April 19--23, 2021, Ljubljana, Slovenia}
\acmPrice{}
\acmDOI{10.1145/3442381.3449887}
\acmISBN{978-1-4503-8312-7/21/04}
\begin{document}

\title{Strongly Local Hypergraph Diffusions for Clustering and Semi-supervised Learning}

\author{Meng Liu}
\email{liu1740@purdue.edu}
\affiliation{%
  \institution{Purdue University}
  \country{United States}
}

\author{Nate Veldt}
\email{nveldt@cornell.edu}
\affiliation{%
  \institution{Cornell University}
  \country{United States}
}

\author{Haoyu Song}
\email{song522@purdue.edu}
\affiliation{%
  \institution{Purdue University}
  \country{United States}
}

\author{Pan Li}
\email{panli@purdue.edu}
\affiliation{%
  \institution{Purdue University}
  \country{United States}
}

\author{David F. Gleich}
\email{dgleich@purdue.edu}
\affiliation{%
  \institution{Purdue University}
  \country{United States}
}

\renewcommand{\shortauthors}{Liu, et al.}

\begin{abstract}
Hypergraph-based machine learning methods are now widely recognized as important for modeling and using higher-order and multiway relationships between data objects.
Local hypergraph clustering and semi-supervised learning specifically involve finding a well-connected set of nodes near a given set of labeled vertices.
Although many methods for local clustering exist for graphs, there are relatively few for localized clustering in hypergraphs.
Moreover, those that exist often lack flexibility to model a general class of hypergraph cut functions or cannot scale to large problems. 
To tackle these issues, this paper proposes a new diffusion-based hypergraph clustering algorithm that solves a quadratic hypergraph cut based objective akin to a hypergraph analog of Andersen-Chung-Lang personalized PageRank clustering for graphs. 
We prove that, for graphs with fixed maximum hyperedge size, this method is strongly local, meaning that its runtime only depends on the size of the output instead of the size of the hypergraph and is highly scalable. Moreover, our method enables us to compute with a wide variety of cardinality-based hypergraph cut functions. 
  We also prove that the clusters found by solving the new objective function satisfy a Cheeger-like quality guarantee.
  We demonstrate that on large real-world hypergraphs our new method finds better clusters and runs much faster than existing approaches. Specifically, it runs in a few seconds for hypergraphs with a few million hyperedges compared with minutes for a flow-based technique.
  We furthermore show that our framework is general enough that can also be used to solve other p-norm based cut objectives on hypergraphs.
  
\end{abstract}





\maketitle
\section{Introduction}
\label{sec:intro}

Two common scenarios in graph-based data analysis are: (i) What are the clusters, groups, modules, or communities in a graph? and (ii) Given some limited label information on the nodes of the graph, what can be inferred about missing labels?  
These statements correspond to the clustering and semi-supervised learning problems respectively, and while there exists a strong state of the art in algorithms for these problems on graphs~\cite{Zhou-2003-semi-supervised,andersen2006local,Gleich-2015-robustifying,yang2020p,Veldt-2016-simple-local-flow,Veldt-2019-flow,Joachims-2003-transductive,Mahoney-2012-local}, research on these problems is currently highly active for hypergraphs~\cite{yoshida2019cheeger,panli_submodular,veldt2020minimizing,Ibrahim-preprint-hocrd,Yin-2017-local-motif,chitra2019random,Zhang2017subgrad} building on new types of results~\cite{Hein2013,panli2017inhomogeneous,veldt2020hypercuts} compared to prior approaches~\cite{karypis1999multilevel,Zhou2006learning,Agarwal2006holearning}. The lack of flexible, diverse, and scalable hypergraph algorithms for these problems limits the opportunities to investigate rich structure in data. For example, clusters can be relevant treatment groups for statistical testing on networks~\cite{Eckles-2017-network-inference} or identify common structure across many types of sparse networks~\cite{Leskovec-2009-community-structure}. Likewise, semi-supervised learning helps to characterize subtle structure in the emissions spectra of galaxies in astronomy data through characterizations in terms of biased eigenvectors~\cite{lawlor2016mapping}. The current set of hypergraph algorithms are insufficient for such advanced scenarios. 

Hypergraphs, indeed, enable a flexible and rich data model that has the potential to capture subtle insights that are difficult or impossible to find with traditional graph-based analysis~\cite{Benson-2016-motif-spectral,Yin-2017-local-motif,panli_submodular,veldt2020minimizing,yoshida2019cheeger}. But, hypergraph generalizations of graph-based algorithms often struggle with scalability and interpretation~\cite{Agarwal2006holearning,Hein2013} with ongoing questions of whether particular models \emph{capture} the higher-order information in hypergraphs.  Regarding scalablility, an important special case for that is a \emph{strongly local algorithm}. Strongly local algorithms are those whose runtime depends on the size of the output rather than the size of the graph. This was only recently addressed for various hypergraph clustering and semi-supervised learning frameworks~\cite{veldt2020minimizing,Ibrahim-preprint-hocrd}. This property enables fast (seconds to minutes) evaluation even for massive graphs with hundreds of millions of nodes and edges~\cite{andersen2006local} (compared with hours). For graphs, perhaps the best known strongly local algorithm is the Andersen-Chung-Lang (henceforth, ACL) approximation for personalized PageRank~\cite{andersen2006local} with applications to local community detection and semi-supervised learning~\cite{Zhou-2003-semi-supervised}. \emph{The specific problem we address is a mincut-inspired hypergraph generalization of personalized PageRank along with a strongly local algorithm to rapidly approximate solutions.} Our formulation differs in a number of important ways from existing Laplacian~\cite{Zhou2006learning} and quadratic function-based hypergraph PageRank generalizations~\cite{panli_submodular,panli2020qdsfm,takai2020hypergraph}.

Although our localized hypergraph PageRank is reasonably simple to state formally (\S\ref{sec:lhpr}), there are a variety of subtle aspects to both the problem statement and the algorithmic solution. First, we wish to have a formulation that enables the richness of possible hypergraph cut functions. These hypergraph cut functions are an essential component to rich hypergraph models because they determine when a group of nodes ought to belong to the same cluster or obtain a potential new label for semi-supervised learning.  Existing techniques based on star expansions (akin to treating the hypergraph as a bipartite graph) or a clique expansion (creating a weighted graph by adding edges from a clique to the graph for each hyperedge) only model a limited set of cut functions~\cite{Agarwal2005beyond,veldt2020hypercuts}.  More general techniques based on Lov\'{a}sz extensions~\cite{panli_submodular,panli2020qdsfm,yoshida2019cheeger} pose substantial computational difficulties.  Second, we need a problem framework that gives sparse solutions such that they can be computed in a strongly local fashion and then we need an algorithm that is actually able to compute these---the mere existence of solutions is insufficient for deploying these ideas in practice as we wish to do. Finally, we need an understanding of the relationship between the results of this algorithm and various graph quantities, such as minimal conductance sets as in the original ACL method.

To address these challenges, we extend and employ a number of recently proposed hypergraph frameworks. First, we show a new result on a class of hypergraph to graph transformations~\cite{veldt2020hypercuts}. These transformations employ carefully constructed directed graph gadgets, along with a set of auxiliary nodes, to encode the properties of a general class of cardinality based hypergraph cut functions. Our simple new result highlights how these transformations not only preserve \emph{cut} values, but preserve the \emph{hypergraph conductance} values as well (\S\ref{sec:hypergraph-to-graph}). Then we localize the computation in the reduced graph using a general strategy to build strongly local computations. This involves a particular modification often called a ``localized cut'' graph or hypergraph~\cite{liu2020strongly,veldt2020minimizing,Andersen:2008:AIG:1347082.1347154,Fountoulakis-preprint-flow}. We then use a squared $2$-norm (i.e. a \emph{quadratic} function) instead of a $1$-norm that arises in the mincut-graph to produce the hypergraph analogue to strongly local personalized PageRank. Put another way, applying all of these steps on a graph (instead of a hypergraph) is equivalent to a characterization of  personalized PageRank vector~\cite{gleich2014anti}. 

Once we have the framework in place (\S\ref{sec:hypergraph-to-graph},\S\ref{sec:lhpr}), we are able to show that an adaptation of the \emph{push} method for personalized PageRank (\S\ref{sec:local-solver}) will compute an approximate solution in time that depends only on the localization parameters and is independent of the size of a hypergraph with fixed maximum hyperedge size (Theorem~\ref{runtime-guarantee}). Consequently, the algorithms are strongly local. 

The final algorithm we produce is extremely efficient. It is a small factor (2-5x) slower than running the ACL algorithm for graphs on the star expansion of the hypergraph. It is also a small factor (2-5x) faster than running an optimized implementation of the ACL algorithm on the clique expansion of the hypergraph. Nevertheless, for many instances of semi-supervised learning problems, it produces results with much larger F1 scores than alternative methods. In particular, it is much faster and performs much better with extremely limited label information than a recently proposed flow-based method~\cite{veldt2020minimizing}.

\mypara{Summary of additional contributions.} In addition to providing a strongly local algorithm for the squared $2$-norm (i.e. a \emph{quadratic} function) in \S\ref{sec:lhpr}, which gives better and faster empirical performance (\S\ref{sec:exp}), we also discuss how to use a $p$-norm (\S\ref{sec:pnorms}) instead. Finally, we also show a Cheeger inequality that relates our results to the hypergraph conductance of a nearby set (\S\ref{sec:theory}).

Our method is the first algorithm for hypergraph clustering that includes \emph{all} of the following features: it is
(1) strongly-local,
(2) can grow a cluster from a small seed set,
(3) models flexible hyperedge cut penalties, and
(4) comes with a conductance guarantee.

\begin{figure*}[t]
\centering
\begin{minipage}{0.19\linewidth}%
\includegraphics[width=0.9\linewidth]{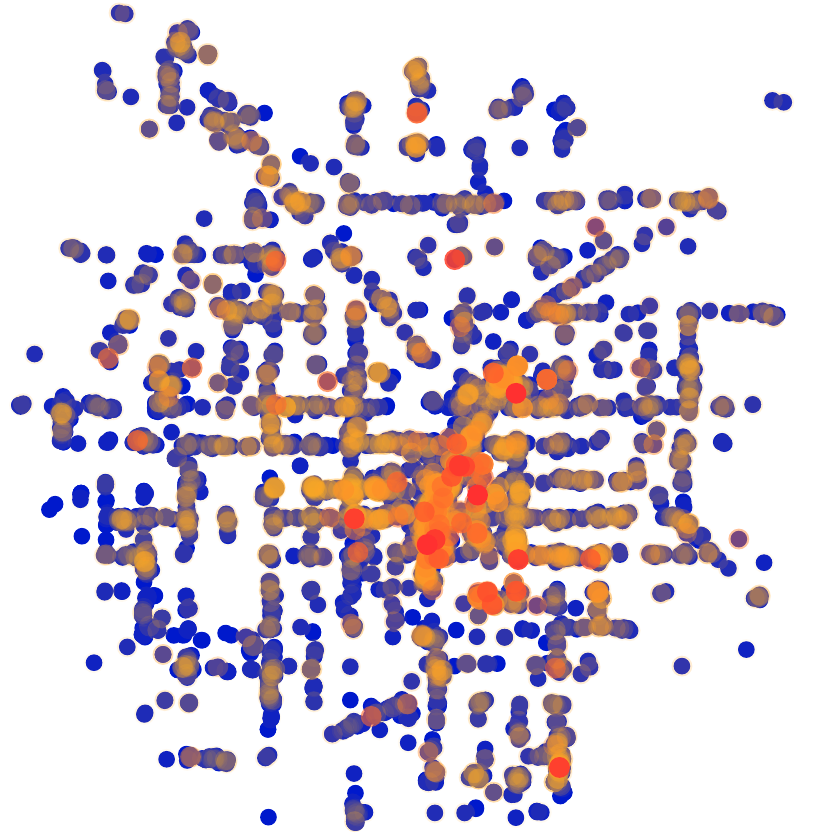}\\
\footnotesize ACL-Clique \\
 P=0.80, R=0.99, F1=0.876
\end{minipage}%
\begin{minipage}{0.19\linewidth}%
\includegraphics[width=0.9\linewidth]{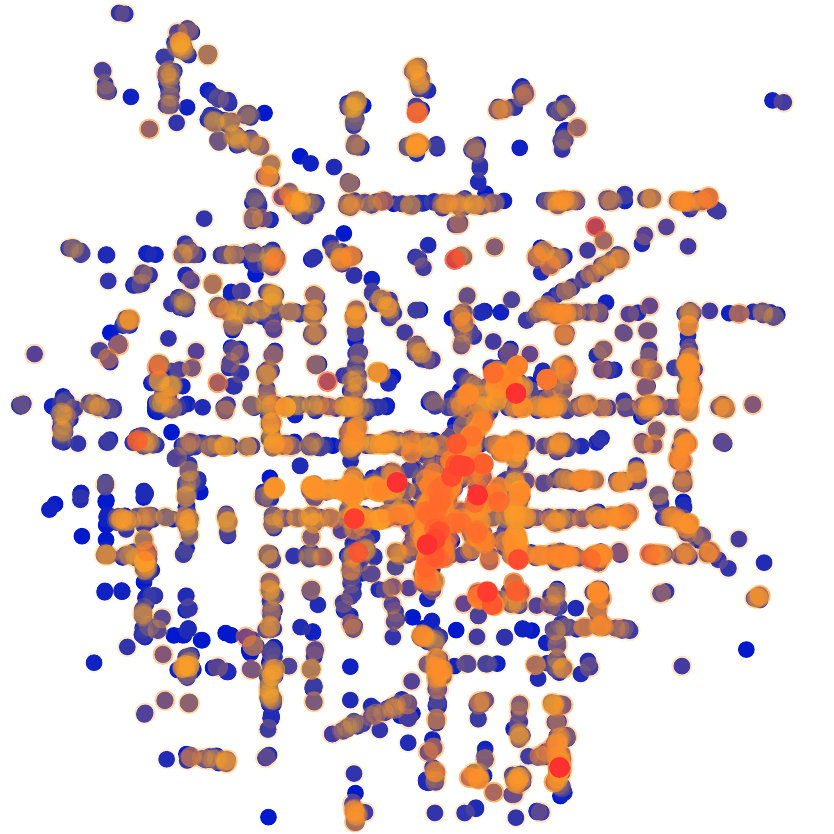}\\
\footnotesize ACL-Star \\
 P=0.76, R=0.98, F1=0.85
\end{minipage}%
\begin{minipage}{0.19\linewidth}%
\includegraphics[width=0.9\linewidth]{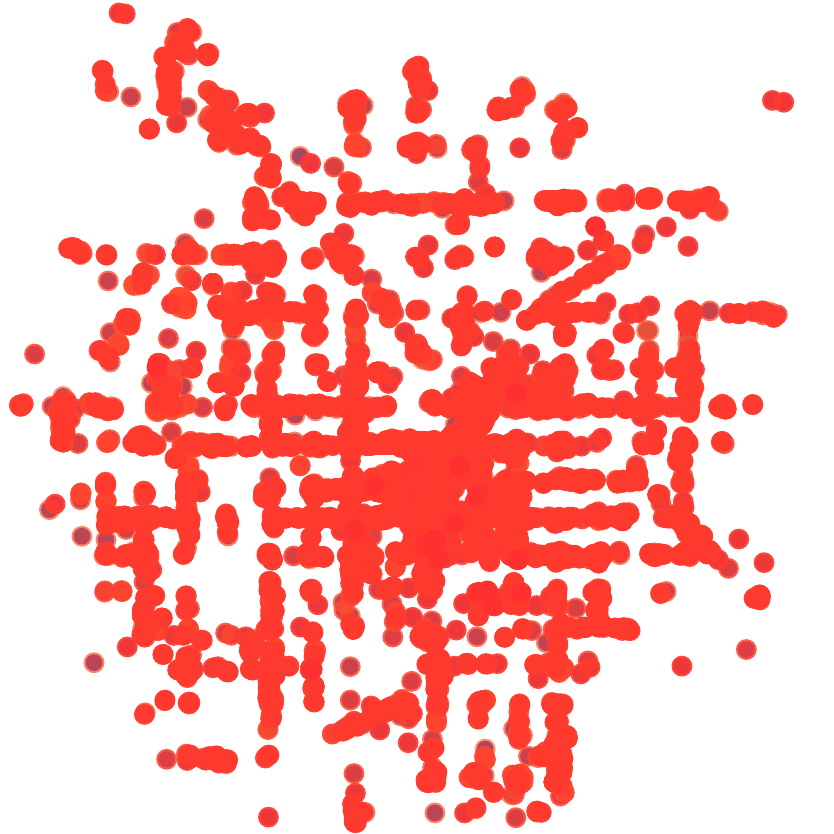}\\
\footnotesize HyperLocal~\cite{veldt2020minimizing} \\
 P=0.92, R=0.05, F1=0.10
\end{minipage}%
\begin{minipage}{0.19\linewidth}%
\includegraphics[width=0.9\linewidth]{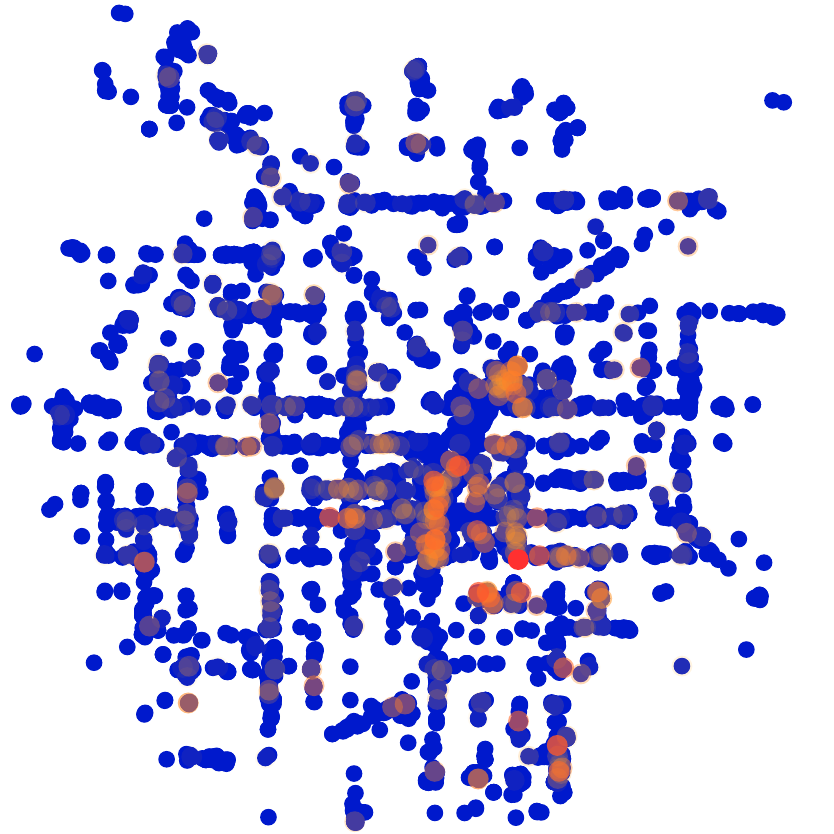}\\
\footnotesize QHPR~\cite{takai2020hypergraph,panli2020qdsfm} \\
 P=0.83, R=0.95, F1=0.886
\end{minipage}%
\begin{minipage}{0.19\linewidth}%
\includegraphics[width=0.9\linewidth]{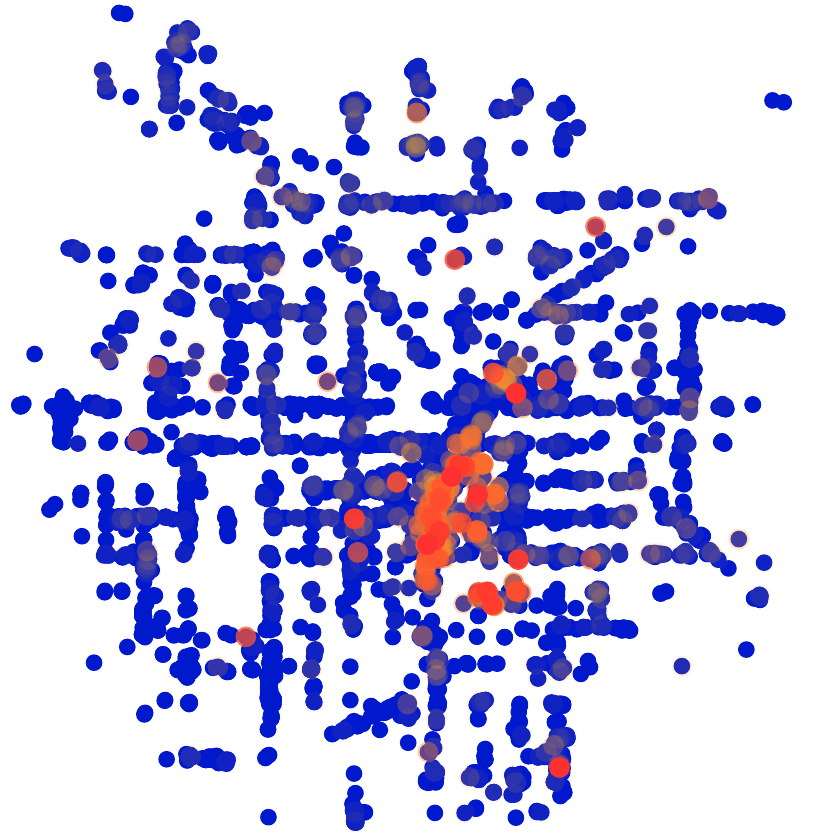}\\
\footnotesize \bfseries \itshape LHPR (Ours)\\
 P=0.83, R=0.98, F1=0.900
\end{minipage}%

\caption{This figure shows locations of the $\sim$7,300 restaurants of Las Vegas that are reviewed on Yelp and how often algorithms recover them from a set of 10 random seeds; our hypergraph PageRank (LHPR) methods has the highest accuracy and finds the result by exploring only 10000 vertices total compared with a fully dense vector for QHPR giving a boost to scalability on larger graphs. The colors show the regions that are missed (red or orange) or found (blue) by each algorithm over 15 trials. HyperLocal is a flow-based method that is known to have trouble growing small seed sets as in this experiment. (The parameters for HyperLocal were chosen in consultation its authors; other parameters were hand tuned for best case performance.)  }
\label{fig:yelp}
\vspace{-10pt}
\end{figure*}

\mypara{A motivating case study with Yelp reviews.} 
We begin by illustrating the need and utility for the methods instead with a simple example of the benefit to these spectral or PageRank-style hypergraph approaches. For this purpose we consider a hypothetical use case with an answer that is easy to understand in order to compare our algorithm to a variety of other approaches. We build a hypergraph from the Yelp review dataset (\url{https://www.yelp.com/dataset}). Each restaurant is a vertex and each user is a hyperedge. This model enables users, i.e.~hyperedges, to capture subtle socioeconomic status information as well as culinary preferences in terms of which types of restaurants they visit and review. The task we seek to understand is either an instance of local clustering or semi-supervised learning. Simply put, given a random sample of 10 restaurants in Las Vegas Nevada, we seek to find other restaurants in Las Vegas. The overall hypergraph has around 64k vertices and 616k hyperedges with a maximum hyperedge size of 2566. Las Vegas, with around 7.3k restaurants, constitutes a small localized cluster.

We investigate a series of different algorithms that will identify a cluster nearby a seed node in a hypergraph: (1)  Andersen-Chung-Lang PageRank on the star and clique expansion of the hypergraph (ACL-Star, ACL-Clique, respectively), these algorithms are closely related to ideas proposed in~\cite{Zhou2006learning,Agarwal2006holearning}; (2) HyperLocal, a recent maximum flow-based hypergraph clustering algorithm~\cite{veldt2020minimizing}; (3) quadratic hypergraph PageRank~\cite{panli2020qdsfm,takai2020hypergraph} (which is also closely related to~\cite{Hein2013}), and (4) our Local Hypergraph-PageRank (LHPR). These are all strongly local except for (3), which we include because our algorithm LHPR is essentially the strongly local analogue of (3). 

The results are shown in Figure~\ref{fig:yelp}. The flow-based HyperLocal method has difficulty finding the entire cluster. Flow-based methods are known to have trouble expanding small seed sets~\cite{Veldt-2016-simple-local-flow,Fountoulakis-preprint-flow,liu2020strongly} and this experiment shows that same behavior. Our strongly local hypergraph PageRank (LHPR) slightly improves on the performance of a quadratic hypergraph PageRank (QHPR) that is not strongly local. In particular, it has 10k non-zero entries (of 64k) in its solution.

This experiment shows the opportunities with our approach for large hypergraphs. 
We are able to model a flexible family of hypergraph cut functions beyond those that use clique and star expansions and we equal or outperform all the other methods. For instance, another more complicated method~\cite{Ibrahim-preprint-hocrd} designed for small hyperedge sizes showed similar performance to ACL-Clique  (F1 around 0.85) and took much longer. 



\vspace{-\baselineskip}
\section{Notation and Preliminaries}
Let $G=(V,E,w)$ be a directed graph with $|V|=n$ and $|E|=m$. For simplicity, we require weights $w_{ij}\geq 1$ for each directed edge $(i,j)\in E$. We interpret an undirected graph as having two directed edges $(i,j)$ and $(j,i)$. For simplicity, we assume the vertices are labeled with indices $1$ to $n$, so that we may use these labels to index matrices and vectors. For instance, we define $\vd$ as the length-$n$ out-degree vector where its $i$th component $d_i=\sum_{j\in V} w_{ij}$. The incidence matrix $\mB\in \{0,-1,1\}^{m\times n}$ measures the difference of adjacent nodes. The $k$th row of $\mB$ corresponds to an edge,  say $(i,j)$, and has exactly two nonzero values, 1 for the node $i$ and -1 for the node $j$. (Recall that we have directed edges, so the origin of the edge is always $1$ and the destination is always -1.) 


Let $\mathcal{H}=(V,\mathcal{E})$ be a hypergraph where each hyperedge $e\in \mathcal{E}$ is a subset of $V$. Let $\zeta =\max_{e\in \mathcal{E}} |e|$ be the maximum hyperedge size. 
With each hyperedge, we associate a \emph{splitting function} $f_e$ that we use to assess an appropriate penalty for splitting the hyperedge among two labels or splitting the hyperedge between two clusters. 
Formally, let $S$ be a cluster and let $A = e \cap S$ be the hyperedge's nodes inside $S$, then $f_e(A)$ penalizes splitting $e$. 
A common choice in early hypergraph literature was the \emph{all-or-nothing} split, which assigns a fixed value if a hyperedge is split or zero if all nodes in the hyperedge lie in the same cluster \cite{hadley1995,ihler1993modeling,lawler1973}: $f_e(A) = 0$ if $A = e$ or $A = \emptyset$ and $f_e(A) = 1$ otherwise (or an alternative constant).
More recently, there have been a variety of alternative splitting functions proposed~\cite{panli_submodular,panli2017inhomogeneous,veldt2020hypercuts} that provide more flexibility. We discuss more choices in the next section (\S\ref{sec:hypergraph-to-graph}). 
With a splitting function identified, the cut value of any given set $S$ can be written as $\text{cut}_{\mathcal{H}}(S)=\sum_{e\in E}f_e(e\cap S)$. The node degree in this case can be defined as $d_i= \sum_{e:i\in e} f_e(\{i\})$~\cite{panli2017inhomogeneous,veldt2020minimizing}, though other types of degree vectors can also be used in both the graph and hypergraph case.
This gives rise to a definition of conductance on a hypergraph
\begin{equation}\label{eq:h-cond}
\phi_{\mathcal{H}}(S)=\frac{\cut_{\mathcal{H}}(S)}{\text{min}(\vol(S),\vol(\bar{S}))}
\end{equation}
where $\vol(S)=\sum_{i\in S}d_i$. This reduces to the standard definition of graph conductance when each edge has only two nodes ($\zeta=2$) and we use the \emph{all-or-nothing penality}.


\mypara{Diffusion algorithms for semi-supervised learning and local clustering.} 
Given a set of seeds, or what we commonly think of as a reference, set $R$, a diffusion is any method that produces a real-valued vector $\vx$ over all the other vertices. For instance, the personalized PageRank method uses $R$ to define the \emph{personalization vector} or \emph{restart vector} underlying the process~\cite{andersen2006local}. The PageRank solution or the sparse Andersen-Chung-Lang approximation~\cite{andersen2006local} are then the diffusion $\vx$. Given a diffusion vector $\vx$, we \emph{round} it back to a set $S$ by performing a procedure called a sweepcut. This involves sorting $\vx$ from largest to smallest and then evaluating the hypergraph conductance of each prefix set $S_j = \{ [1], [2], \ldots, [k] \}$, where $[i]$ is the id of the $i$th largest vertex. The set returned by sweepcut picks the minimum conductance set $S_j$. Since the sweepcut procedures are general and standardized, we focus on the computation of $\vx$.  When these algorithms are used for semi-supervised learning, the returned set $S$ is presumed to share the label as the reference (seed) set $R$; alternatively, its value or rank information may be used to disambiguate multiple labels~\cite{Zhou-2003-semi-supervised,Gleich-2015-robustifying}.

\section{Method}
Our overall goal is to compute a hypergraph diffusion that will help us perform a sweepcut to identify a set with \emph{reasonably small} conductance nearby a reference set of vertices in the graph. We explain our method: \emph{localized hypergraph quadratic diffusions} (LHQD) or also \emph{localized hypergraph PageRank} (LHPR) through two transformations before we formally state the problem and algorithm. We adopted this strategy so that the final proposal is well justified because some of the transformations require additional context to appreciate. Computing the final sweepcut is straightforward for hypergraph conductance, and so we do not focus on that step. 

\subsection{Hypergraph-to-graph reductions }
\label{sec:hypergraph-to-graph}


Minimizing conductance is NP-hard even in the case of simple graphs, though numerous techniques have been designed to approximate the objective in theory and practice~\cite{Andersen:2008:AIG:1347082.1347154,andersen2006local,Chung-1992-book}. A common strategy for searching for low-conductance sets in hypergraphs is to first reduce a hypergraph to a graph, and then apply existing graph-based techniques. This sounds ``hacky'' or least ``ad-hoc'' but this idea is both principled and rigorous. The most common approach is to apply a clique expansion~\cite{Benson-2016-motif-spectral,panli2017inhomogeneous,zien1999,Zhou2006learning,Agarwal2006holearning}, which explicitly models splitting functions of the form 
$f_e(A)\propto |A||e\backslash A|$. 
For instance Benson et al.~\cite{Benson-2016-motif-spectral} showed that clique expansion can be used to convert a 3-uniform hypergraph into a graph that preserves the \emph{all-or-nothing} conductance values. For larger hyperedge sizes, \emph{all-or-nothing} conductance is preserved to within a distortion factor depending on the size of the hyperedge. Later, Li et al.~\cite{panli2017inhomogeneous} were the first to introduce more generalized notions of hyperedge splitting functions, focusing specifically on submodular functions.
\begin{definition}
A splitting function $f_e$ is \textit{submodular} if  
\begin{equation}\label{eq:subm}
f_e(A)+f_e(B)\geq f_e(A\cup B) + f_e(A\cap B)\quad\text{$\forall A,B\subseteq e$}.
\end{equation}
\end{definition}
These authors showed that for this submodular case, clique expansion could be used to define a graph preserving conductance to within a factor $O(\zeta)$ ($\zeta$ is the largest hyperedge size).

More recently, Veldt et al.~\cite{veldt2020hypercuts} introduced graph reduction techniques that \emph{exactly} preserve submodular hypergraph \emph{cut} functions which are cardinality-based. 
\begin{definition}
A splitting function $f_e$ is cardinality-based if  
\begin{equation}\label{eq:card}
f_e(A) = f_e(B)\quad\text{whenever $|A| = |B|$}.
\end{equation}
\end{definition}
Cardinality-based splitting functions are a natural choice for many applications, since node identification is typically irrelevant in practice, and the cardinality-based model produces a cut function that is invariant to node permutation. Furthermore, most previous research on applying generalized hypergraph cut penalties implicitly focused on cut functions are that are naturally cardinality-based~\cite{panli_submodular,panli2020qdsfm,Hein2013,Benson-2016-motif-spectral,zien1999,karypis1999multilevel}. 
Because of their ubiquity and flexibility, in this  work we also focus on hypergraph cut functions that are submodular and cardinality-based. We briefly review the associated graph transformation and then we build on previous work by showing that these hypergraph reductions can be used to preserve the hypergraph \emph{conductance} objective, and not just hypergraph cuts.  

\mypara{Reduction for Cardinality-Based Cuts.}
Veldt et al.~\cite{veldt2020hypercuts} gave results that show the cut properties of a submodular, cardinality-based hypergraph could be preserved by replacing each hyperedge with a set of directed graph \emph{gadgets}. Each gadget for a hyperedge $e$ is constructed by introducing a pair of auxiliary nodes $a$ and $b$, along with a directed edge $(a, b)$ with weight $\delta_e > 0$. For each $v \in e$, two unit-weight directed edges are introduced: $(v, a)$ and $(b, v)$. The entire gadget is then scaled by a weight $c_e \geq 0$. The resulting gadget represents a simplified splitting function of the following form:
\begin{equation}
\label{cb-function}
    f_e(A) = c_e \cdot \min \{ |A|, |e\backslash A|, \delta_e \}.
\end{equation}
Figure~\ref{example-graph-reduction}(b) illustrates the process of replacing a hyperedge with a gadget. The cut properties of any submodular cardinality-based splitting function can be exactly modeled by introducing a set of $O(|e|)$ or fewer such splitting functions~\cite{veldt2020hypercuts}. If an approximation suffices, only $O(\log |e|)$ gadgets are required~\cite{benson2020augmented}.

\begin{figure*}
\begin{minipage}{0.24\linewidth}
  \includegraphics[width=0.95\linewidth]{"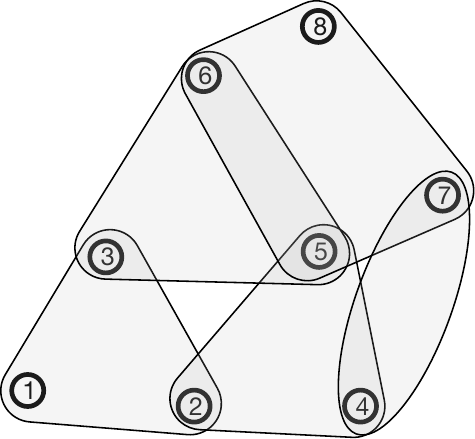"}\\
  \footnotesize (a) original hypergraph\\
\end{minipage}\hfill%
\begin{minipage}{0.12\linewidth}
\includegraphics[width=0.85\linewidth]{"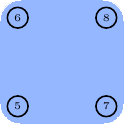"}\\
\includegraphics[width=0.85\linewidth]{"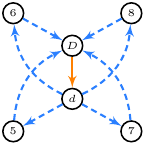"}\\
\footnotesize (b) single hyperedge reduction gadget
\end{minipage}\hfill%
\begin{minipage}{0.24\linewidth}
  \includegraphics[width=0.95\linewidth]{"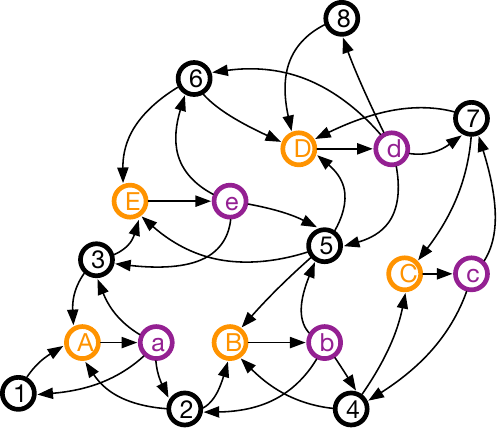"} \\
  \footnotesize (c) expanded graph\\
\end{minipage}\hfill%
\begin{minipage}{0.28\linewidth}
  \includegraphics[width=0.95\linewidth]{"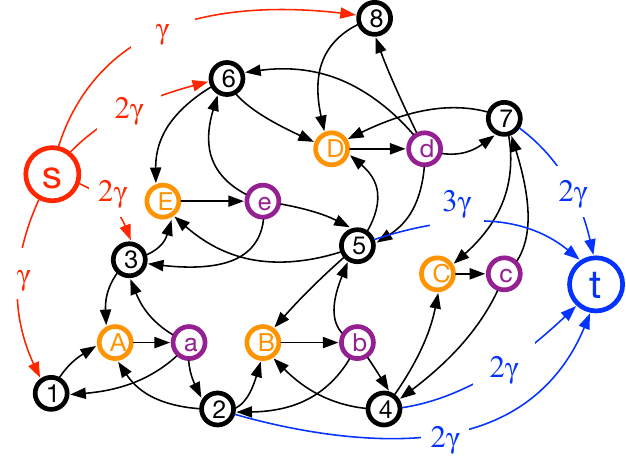"} \\
  \footnotesize (d) localized directed cut graph\\
\end{minipage}
\vspace{-\baselineskip}
\caption{A simple illustration of hypergraph reduction (Section~\ref{sec:hypergraph-to-graph}) and localization (Section~\ref{sec:lhqd}). (a) A hypergraph with 8 nodes and 5 hyperedges. (b) An illustration of the hyperedge transformation gadget for $\delta$-linear splitting function. (c) The hypergraph is reduced to a directed graph by adding a pair of auxiliary nodes for each hyperedge and this preserves hypergraph conductance computations (Theorem~\ref{thm:preservecond}). (d) The localized directed cut graph is created by adding a source node $s$, a sink node $t$ and edges from $s$ to hypergraph nodes or from hypergraph nodes to $t$ to \textit{localize} a solution.}
\label{fig:transformations}
\label{example-graph-reduction}
\label{full-graph-reduction}
\vspace{-10pt}
\end{figure*}


An important consequence of these reduction results is that in order to develop reduction techniques for \emph{any} submodular cardinality-based splitting functions, it suffices to consider hyperedges with splitting functions of the simplified form given in \eqref{cb-function}. In the remainder of the text, we focus on splitting functions of this form, with the understanding that all other cardinality-based submodular splitting functions can be modeled by introducing multiple hyperedges on the same set of nodes with different edge weights.

In Figure~\ref{full-graph-reduction}, we illustrate the procedure of reducing a small hypergraph to a directed graph, where we introduce a single gadget per hyperedge. Formally, for a hypergraph $\mathcal{H} = (V,E)$, this procedure produces a directed graph $G = (\hat{V}, \hat{E})$, with directed edge set $\hat{E}$, and node set $\hat{V} = V \cup V_a \cup V_b$, where $V$ is the set of original hypergraph nodes. Sets $V_a, V_b$ store auxiliary nodes, in such a way that for each pair of auxiliary nodes $a, b$ where $(a,b)$ is a directed edge, we have $a \in V_a$ and $b \in V_b$. This reduction technique was previously developed as a way of preserving minimum cuts and minimum $s$-$t$ cuts for the original hypergraph. Here, we extend this result to show that for a certain choice for node degree, this reduction also preserves hypergraph conductance.

\begin{theorem}
\label{thm:preservecond}
Define a degree vector $\vd$ for the reduced graph $G = (\hat{V}, \hat{E})$ such that $\vd(v) = d_v$ is the out-degree for each node $v \in V$, and $\vd(u) = d_u = 0$ for every auxiliary node $u \in V_a\cup V_b$. If $T^*$ is the minimum conductance set in $G$ for this degree vector, then $S^* = T^*\cap V$ is the minimum hypergraph conductance set in $\mathcal{H} = (V,E)$. 
\end{theorem}
\begin{proof}
From previous work on these reduction techniques~\cite{benson2020augmented,veldt2020hypercuts}, we know that the cut penalty for a set $S \subseteq V$ in $\mathcal{H}$ equals the cut penalty in the directed graph, as long as auxiliary nodes are arranged in a way that produces the smallest cut penalty subject to the choice of node set $S \subseteq V$. Formally, for $S \subseteq V$,
\begin{equation}
    \cut_\mathcal{H}(S) = \minimize_{T \subset \hat{V}\colon S = T \cap V} \cut_G(T),
    \end{equation}
    where $\cut_G$ denotes the weight of directed out-edges originating inside $S$ that are cut in $G$.
    By our choice of degree vector, the volume of nodes in $G$ equals the volume of the non-auxiliary nodes in $\mathcal{H}$. That is, for all $T \subseteq \hat{V}$,
    $\vol_G(T) = \sum_{v \in V} d_v + \sum_{u \in V_a\cup V_b} d_u = \vol_G(T \cap V) = \vol_\mathcal{H}(T \cap V)$.
Let $T^* \subseteq \hat{V}$ be the minimum conductance set in $G$, and $S^* = T^*\cap V$. Without loss of generality we can assume that $\vol_G(T^*) \leq \vol_G(\bar{T}^*)$. Since $T^*$ minimizes conductance, and auxiliary nodes have no effect on the volume of this set, $\cut_G(T^*) = \minimize_{T\subset \hat{V} \colon T \cap S^*} \cut_G(T) = \cut_\mathcal{H}(S^*)$, and so
${\cut_G(T^*)}/{\vol_G(T^*)} = {\cut_\mathcal{H}(S^*)}/{\vol_\mathcal{H}(S^*)}$.
Thus, minimizing conductance in $G$ minimizes conductance in $\mathcal{H}$.
\end{proof}

\subsection{Localized Quadratic Hypergraph Diffusions}
\label{sec:lhpr}
\label{sec:lhqd}

Having established a conductance-preserving reduction from a hypergraph to a directed graph, we now present a framework for detecting localized clusters in the reduced graph $G$. To accomplish this, we first define a \emph{localized directed cut graph}, involving a source and sink nodes and new weighted edges. This approach is closely related to previously defined localized cut graphs for local graph clustering and semi-supervised learning~\cite{Andersen:2008:AIG:1347082.1347154,gleich2014anti,Veldt-2016-simple-local-flow,liu2020strongly,Zhu-2003-diffusion,Blum-2001-mincuts}, and a similar localized cut hypergraph used for flow-based hypergraph clustering~\cite{veldt2020minimizing}. The key conceptual difference is that we apply this construction directly to the reduced graph $G$, which by Theorem~\ref{thm:preservecond} preserves conductance of the original hypergraph $\mathcal{H}$. Formally, we assume we are given a set of nodes $R \subseteq V$ around which we wish to find low-conductance clusters, and a parameter $\gamma > 0$. The localized directed cut graph is defined by applying the following steps to $G$:
\begin{itemize}
    \item Introduce a source node $s$, and for each $r \in R$ define a directed edge $(s,r)$ of weight $\gamma d_r$
    \item Introduce a sink node $t$, and for each $v \in \bar{R}$ define a directed edge $(v,t)$ with weight $\gamma d_v$.
\end{itemize}
We do not connect auxiliary nodes to the source or sink, which is consistent with the fact that their degree is defined to be zero in order for Theorem~\ref{thm:preservecond} to hold. We illustrate the construction of the localized directed cut graph in Figure~\ref{full-graph-reduction}(d). It is important to note that in practice we do not in fact form this graph and store it in memory. Rather, this provides a conceptual framework for finding localized low-conductance sets in $G$, which in turn correspond to good clusters in $\mathcal{H}$.

\mypara{Definition: Local hypergraph quadratic diffusions.} Let $\mB$ and $\vw$ be the incidence matrix and edge weight vector of the localized directed cut graph with $\gamma$. The objective function for our hypergraph clustering diffusion, which we call \emph{local hypergraph quadratic diffusion} or simply \emph{local hypergraph PageRank}, is
\begin{equation}
\label{hypergraph-2-norm-cut}
\MINone{\vx}{\frac{1}{2}\vw^T (\mB\vx)^2_+ + \kappa \gamma \sum_{i\in V} x_i d_i}{x_s=1,x_t=0,\vx\geq 0.}
\end{equation}	
We use the function $(x)_+ = \max \{x, 0\}$, applied element-wise to $\mB\vx$, to indicate we only keep the positive elements of this product. This is analogous to the fact that we only view a directed edge as being cut if it crosses from the source to the sink side; this is similar to previous directed cut minorants on graphs and hypergraphs~\cite{yoshida2016nonlinear}. 
The first term in the objective corresponds to a 2-norm minorant of the minimum $s$-$t$ cut objective on the localized directed cut graph. (In an undirected regular graph, the term $ \vw^T (\mB \vx)_+$ turns into an expression with the Laplacian, which can in turn be formally related to PageRank~\cite{gleich2014anti}).  If instead, we replace exponent 2 with a 1 and ignore the second term, this would amount to finding a minimum $s$-$t$ cut (which can be solved via a maximum flow). The second term in the objective is included to encourage sparsity in the solution, where $\kappa \ge 0$ controls the desired level of sparsity. With $\kappa > 0$ we are able to show in the next section that we can compute solutions in time that depends only on $\kappa, \gamma,$ and $\vol(R)$, which allows us to evaluate solutions to \eqref{hypergraph-2-norm-cut} in a strongly local fashion.

\subsection{A strongly local solver for LHQD \eqref{hypergraph-2-norm-cut}}
\label{sec:local-ppr}
\label{sec:local-solver}

In this section, we will provide a strongly local algorithm to approximately satisfy the optimality conditions of \eqref{hypergraph-2-norm-cut}.  We first state the optimality conditions in Theorem~\ref{hyper-cut-KKT}, and then present the algorithm to solve them. 
The simplest way to understand this algorithm is as a generalization of the Andersen-Chung-Lang push procedure for PageRank~\cite{andersen2006local}, which we will call ACL as well as the more recent nonlinear push procedure~\cite{liu2020strongly}. Two new challenges about this new algorithm are: (1) the new algorithm operates on a directed graph, which means unlike ACL there is no single closed form update at each iteration and (2) there is no sparsity regularization for auxiliary nodes, which will break the strongly local guarantees for existing analyses of the push procedure. 

We begin with the optimality conditions for \eqref{hypergraph-2-norm-cut}. 
\begin{theorem}
\label{hyper-cut-KKT}
Fix a seed set $R$, $\gamma>0$, $\kappa>0$, define a residual function $\vr(\vx)=-\frac{1}{\gamma}\mB^T\text{diag}((\mB\vx)_+)\vw$. A necessary and sufficient condition to satisfy the KKT conditions of~\eqref{hypergraph-2-norm-cut} is to find $\vx^*$ where $\vx^*\geq 0$, $\vr(\vx^*)=[r_s,\vg^T,r_t]^T$ with $g_i\leq\kappa d_i$ (where $\vd$ reflects the graph before adding $s$ and $t$ but does include the $0$ degree nodes), $(\kappa d_i-g_i)^Tx^*_i=0$ for $i\in V$ and $g_i=0$ for all auxilary nodes added.
\end{theorem}
The proof of this would be included in a longer version of this material; however, we omit the details in the interest of space as it is a straightforward application of determining optimality conditions for convex programs. We further note that solutions $\vx^*$ are unique because the problem is strongly convex due to the quadratic. 


In \S\ref{sec:hypergraph-to-graph}, we have shown that the reduction technique of any cardinality submodular-based splitting function suffices to introduce multiple directed graph gadgets with different $\delta_e$ and $c_e$. In order to simplify our exposition, we assume that each hyperedge has a $\delta$-linear threshold splitting function~\cite{veldt2020minimizing} $f_e=\text{min}\{|A|,|e\backslash A|,\delta\}$ with $\delta\geq 1$ to be a tunable parameter. This splitting function can be exactly modeled by replacing each hyperedge with one directed graph gadget with $c_e=1$ and $\delta_e=\delta$. (This is what is illustrated in Figure~\ref{full-graph-reduction}.) Also when $\delta=1$, it models the standard unweighted \emph{all-or-nothing} cut~\cite{hadley1995,ihler1993modeling,lawler1973} and when $\delta$ goes to infinity, it models star expansion~\cite{zien1999}. Thus this splitting function can interpolate these two common cut objectives on hypergraphs by varying $\delta$.

By assuming that we have a $\delta$-linear threshold splitting function, this means we can associate exactly \emph{two} auxiliary nodes with each hyperedge. We call these $a$ and $b$ for simplicity. We also let $V_a$ be the set of all $a$ auxilary nodes and $V_b$ be the set of all $b$ nodes. 

At a high level, the algorithm to solve this proceeds as follows: whenever there exists a graph node $i\in V$ that violates optimality, i.e. $r_i>\kappa d_i$, we first perform a $\texttt{hyperpush}$ at $i$ to increase $x_i$ so that the optimality condition is approximately satisfied, i.e., $r_i=\rho\kappa d_i$ where $0<\rho<1$ is a given parameter that influences the approximation. This changes the solution $\vx$ only at the current node $i$ and residuals at adjacent auxiliary nodes. Then we immediately \emph{push} on adjacent auxiliary nodes, which means we increase their value so that the residuals remain zero. After pushing each pair $(a,b)$ of associated auxiliary nodes, we then update residuals for \emph{adjacent} nodes in $V$. Then we search for another optimality violation. (See Algorithm~\ref{cut-algo} for a formalization of this strategy.) When $\rho<1$, we only approximately satisfy the optimality conditions; and this approximation strategy has been repeatedly and successfully used in existing local graph clustering algorithms~\cite{andersen2006local,gleich2014anti,liu2020strongly}.

\begin{algorithm}[t]
\caption{$\texttt{LHQD}(\gamma,\kappa,\rho)$ for set $R$ and hypergraph $H$ with $\delta$-linear penalty where $0<\rho<1$ determines accuracy}
\begin{algorithmic}[1]
\label{cut-algo}
\STATE Let $\vx\!=\!0$ except for $x_s\!=\!1$ and set $\vr\!=\!-\gamma^{-1}\mB^T\text{diag}((\mB\vx)_+)\vw(\!\delta,\gamma\!)$. \\[-\baselineskip]
\STATE While there is any vertex $i\in V$ where $r_i > \kappa d_i$, or stop if none exists
            \emph{\textcolor{blue}{(find an optimality violation)}}
\STATE \begin{myindentpar}{1em} 
  Perform $\texttt{LHQD-hyperpush}$ at vertex $i$ so that $r_i=\rho\kappa d_i$, updating $\vx$ and $\vr$. \emph{\textcolor{blue}{(satisfy optimality at i)}}
\end{myindentpar}
\STATE 
\begin{myindentpar}{1em} 
    For each pair of adjacent auxiliary nodes $a$, $b$ 
    where $a\in V_a$, $b\in V_b$ and $a\rightarrow b$, perform $\texttt{LHQD-auxpush}$ at $a$ and $b$ so that $r_a=r_b=0$, then update $\vx$ and $\vr$ after each \texttt{auxpush}. 
\end{myindentpar}
\STATE Return $\vx$
\end{algorithmic}
\end{algorithm}

\mypara{Notes on optimizing the procedure.} 
Algorithm~\ref{cut-algo} formalizes a general strategy to approximately solve these diffusions. We now note a number of optimizations that we have found to greatly accelerate this strategy. First, note that $\vx$ and $\vr$ can be kept as sparse vectors with only a small set of entries stored. Second, note that we can maintain a list of optimality violations because each update to $\vx$ only causes $\vr$ to increase, so we can simply check if each coordinate increase creates a new violation and add it to a queue. Third, to find the value that needs to be ``pushed'' to each node, a general strategy is to use a binary search procedure as we will use for the $p$-norm generalization in \S\ref{sec:pnorms}. 
However, if the tolerance of the binary search is too small, it will slow down each iteration. If the tolerance is too large, the solution will be too far away from the true solution to be useful. In the remaining of this section, we will show that in the case of quadratic objective~\eqref{hypergraph-2-norm-cut}, we can (i) often avoid binary search and (ii) when it is still required, make the binary search procedure unrelated to the choice of tolerance in those iterations where we do need it. These detailed techniques will not change the time complexity of the overall algorithm, but make a large difference in practice. 

We will start by looking at the expanded formulations of the residual vector. When $i\in V$, $r_i$ expands as:
\begin{equation}\label{ri} r_i = \frac{1}{\gamma} \sum_{b\in V_b}w_{bi}(x_b-x_i)_+-\frac{1}{\gamma}\sum_{a\in V_a}w_{ia}(x_i-x_a)_+ + d_i [\text{Ind}(i \in R) - x_i].
\end{equation}
Similarly, for each $a\in V_a$, $b\in V_b$ where $a\rightarrow b$, they will share the same set of original nodes and their residuals can be expanded as:
\begin{equation}\label{ra-rb}
\begin{aligned}
&\textstyle r_a=-w_{ab}(x_a-x_b)+\sum_{i\in V}w_{ia}(x_i-x_a)_+ \\
&\textstyle r_b=w_{ab}(x_a-x_b)-\sum_{i\in V}w_{bi}(x_b-x_i)_+ \\
\end{aligned}
\end{equation}
Note here we use a result that $x_a\geq x_b$ (Lemma~\ref{aux-lemma}). 

The goal in each hyperpush is to first find $\Delta x_i$ such that $r'_i=\rho\kappa d_i$ and then in auxpush, for each pair of adjacent auxiliary nodes $(a,b)$, find $\Delta x_a$ and $\Delta x_b$ such that $r'_a$ and $r'_b$ remain zero. ($\Delta x_i$, $\Delta x_a$ and $\Delta x_b$ are unique because the quadratic is strongly convex.) Observe that $r_i$, $r_a$ and $r_b$ are all piecewise linear functions, which means we can derive a closed form solution once the relative ordering of adjacent nodes is determined. Also, in most cases, the relative ordering won't change after a few initial iterations. So we can first reuse the ordering information from last iteration to directly solve $\Delta x_i$, $\Delta x_a$ and $\Delta x_b$ and then check if the ordering is changed. 

Given these observations, we will record and update the following information for each pushed node. Again, this information can be recorded in a \emph{sparse} fashion. 
When the pushed node $i$ is a original node, for its adjacent $a\in V_a$ and $b\in V_b$, we record:
\begin{itemize}
    \item $s_a^{(i)}$: the sum of edge weights $w_{ia}$ where $x_a<x_i$
    \item $s_b^{(i)}$: the sum of edge weights $w_{bi}$ where $x_b>x_i$
    \item $a_{min}^{(i)}$: the minimum $x_a$ where $x_a\geq x_i$
    \item $b_{min}^{(i)}$: the minimum $x_b$ where $x_b> x_i$
\end{itemize}
Now assume the ordering is the same, $r'_i$ can be written as $r'_i=r_i-\frac{1}{\gamma}(s_a^{(i)}+s_b^{(i)})\Delta x_i=\rho\kappa d_i$, so
\begin{equation}
\label{solve-delta-xi}
\Delta x_i=\gamma(r_i-\rho\kappa d_i)/(s_a^{(i)}+s_b^{(i)}).
\end{equation}
Then we need to check whether the assumption holds by checking 
\begin{equation}
\label{delta-xi-condition}
x_i+\Delta x_i\leq \text{min}\left(a_{min}^{(i)},b_{min}^{(i)}\right)
\end{equation}
If not, we need to use binary search to find the new location of $x_i+\Delta x_i$ (Note $\Delta x_i$ here is the true value that is still unknown), update $s_a^{(i)}$, $s_b^{(i)}$, $a_{min}^{(i)}$ and $b_{min}^{(i)}$ and recompute $\Delta x_i$. This process is summarized in \texttt{LQHD-hyperpush}.
\begin{algorithm}[t]
\caption{$\texttt{LQHD-hyperpush}(i,\gamma,\kappa,\vx,\vr,\rho)$}
\begin{algorithmic}[1]
\label{fast-push-algo}
\STATE Solve $\Delta x_i$ with $s_a^{(i)}$, $s_b^{(i)}$, $a_{min}^{(i)}$ and $b_{min}^{(i)}$ using \eqref{solve-delta-xi}. \emph{\textcolor{blue}{(assume the order of i doesn't change among its adjacent nodes)}}
\IF {\eqref{delta-xi-condition} doesn't hold \emph{\textcolor{blue}{(adding $\Delta x_i$ changed the order of i)}}}
\STATE Binary search on $\Delta x_i$ until we find the smallest interval among all adjacent nodes of $i$ that will include $x_i +\Delta x_i$, update $s_a^{(i)}$, $s_b^{(i)}$, $a_{min}^{(i)}$ and $b_{min}^{(i)}$.
\STATE Solve $\Delta x_i$ with the found interval by setting $r_i=\rho\kappa d_i$ in~\eqref{ri}.
\ENDIF 
\STATE Update $\vx$ and $\vr$, $x_i \leftarrow x_i+\Delta x_i$, $r_i\leftarrow\rho\kappa d_i$
\end{algorithmic}
\end{algorithm}

Similarly, when the pushed nodes $a\in V_a$, $b\in V_b$ where $a\rightarrow b$, are a pair of auxiliary nodes, for its adjacent nodes $i\in V$, we record:
\begin{itemize}
    \item $z_a$: the sum of edge weights $w_{ia}$ where $x_a<x_i$
    \item $z_b$: the sum of edge weights $w_{bi}$ where $x_b>x_i$
    \item $x_{min}^{(a)}$: the minimum $x_i$ where $x_a<x_i$
    \item $x_{min}^{(b)}$: the minimum $x_i$ where $x_b<x_i$
\end{itemize}
Then we solve $\Delta x_a$, $\Delta x_b$ by solving the following linear system (here we assume $x_b\geq x_i$)
\begin{equation}
\label{solve-delta-xab}
\left\{\begin{aligned}
    -w_{ab}(\Delta x_a-\Delta x_b)+\frac{w_{ia}}{\gamma}((x'_i-x_a)_+-(x_i-x_a)_+)-z_a\Delta x_a=0\\
    w_{ab}(\Delta x_a-\Delta x_b)-\frac{w_{bi}}{\gamma}((x_b-x_i')_+-(x_b-x_i)_+)+z_b\Delta x_b=0
\end{aligned}\right.
\end{equation}
where $x_i'$ refers to the updated $x_i$ after applying \texttt{LQHD-hyperpush} at node $i$. And the assumption will hold if and only if the following inequalities are all satisfied:
\begin{equation}
\label{delta-xab-condition}
     x'_i\leq x_b, \qquad 
     x_a+\Delta x_a\leq x_{min}^{(a)}, \qquad 
     x_b+\Delta x_b\leq x_{min}^{(b)} 
\end{equation}
If not, we also need to use binary search to update the locations of $x_a+\Delta x_a$ and $x_b+\Delta x_b$, update $z_a$, $z_b$, $x_{min}^{(a)}$, $x_{min}^{(b)}$ and recompute $\Delta x_a$ and $\Delta x_b$.

\begin{algorithm}[t] 
\caption{$\texttt{LQHD-auxpush}(i,a,b,\gamma,\vx,\vr,\Delta x_i)$}
\begin{algorithmic}[1]
\label{fast-aux-push-algo}
\STATE Solve $\Delta x_a$, $\Delta x_b$ with $z_a$, $z_b$, $x_{min}^{(a)}$ and $x_{min}^{(b)}$ using \eqref{solve-delta-xab}. 
\IF {\eqref{delta-xab-condition} doesn't hold. \emph{\textcolor{blue}{(adding $\Delta x_a, \Delta x_b$ altered the order)}}}
\STATE Binary search on $\Delta x_a$ until we find the smallest interval among all adjacent original nodes of $a$ that will include $x_a+\Delta x_a$, update $z_a$, $x_{min}^{(a)}$, similarly for $z_b$, $x_{min}^{(b)}$. 
\STATE Solve $\Delta x_a, \Delta x_b$ with the found intervals by setting $r_a=r_b=0$ in~\eqref{ra-rb}.
\ENDIF
\STATE Change the following entries in $\vx$ and $\vr$ to update the solution and the residual
\STATE (a) $x_a\leftarrow x_a\!+\Delta x_a$ and $x_b\leftarrow x_b\!+\Delta x_b$
\STATE (b) For each neighboring node $i\rightarrow a$ where $i\in V$, $r_i \leftarrow r_i\!+\!\frac{1}{\gamma}w_{ia}(x_i\!-\!x_a)_+\!-\!\frac{1}{\gamma}w_{ia}(x_i\!-x_a\!-\!\Delta x_a)_+-\!\frac{1}{\gamma}w_{bi}(x_b\!-x_i)_+\!+\!\frac{1}{\gamma}w_{bi}(x_b\!+\!\Delta x_b\!-x_i)_+$
\end{algorithmic}
\end{algorithm}

\mypara{Establishing a runtime bound.}
The key to understand the strong locality of the algorithm is that after each \texttt{LQHD-hyperpush}, the decrease of $\|vg\|_1$ can be lower bounded by a value that is independent of the total size of the hypergraph, while \texttt{LHQD-auxpush} doesn't change $\|\vg\|_1$. Formally, we have the following theorem:
\begin{theorem}
\label{runtime-guarantee}
Given $\gamma>0$, $\kappa>0$, $\delta>0$ and $0<\rho<1$. Suppose the splitting function $f_e$ is submodular, cardinality-based and satisfies $1 \le f_e(\{i\})\leq\delta$ for any $i\in e$. Then calling \texttt{LQHD-auxpush} doesn't change $\| \vg \|_1$ while calling \texttt{LQHD-hyperpush} on node $i\in V$ will decrease $\|\vg\|_1$ by at least $\gamma\kappa(1-\rho)d_i/(\gamma\kappa+\delta)$. 

Suppose $\texttt{LHQD}$ stops after $T$ iterations and $d_i$ is the degree of the original node updated at the $i$-th iteration, then $T$ must satisfy:
\[\textstyle\sum_{i=1}^Td_i\leq(\gamma\kappa+\delta)\text{vol}(R)/\gamma\kappa(1-\rho)=O(\text{vol}(R)).\]
\end{theorem}
The proof is in the appendix.
Note that this theorem only upper bounds the number of iterations Algorithm~\ref{cut-algo} requires. Each iteration of Algorithm~\ref{cut-algo} will also take $O(\sum_{e\in\mathcal{E}, i\in e}|e|)$ amount of work. This ignores the binary search, which only scales it by $\text{log}(\text{max}\{d_i,\text{max}_{e\in\mathcal{E}, i\in e}\{|e|\}\})$ factor in the worst case. Putting these pieces together shows that if we have a hypergraph with \emph{independently bounded} maximum hyperedge size, then we can treat this additional work as a constant. Consequently, our solver is strongly local for graphs with bounded maximum hyperedge size; this matches the interpretation in~\cite{veldt2020minimizing}.

\section{Local conductance approximation}
\label{sec:theory}
We give a local conductance guarantee that results from solving~\eqref{hypergraph-2-norm-cut}. Because of space, we focus on the case $\kappa=0$.
We prove that a sweepcut on the solution $\vx$ of \eqref{hypergraph-2-norm-cut} leads to a Cheeger-type guarantee for   conductance of the hypergraph $\mathcal{H}$ even when the seed-set size $|R|$ is $1$. It is extremely difficult to guarantee a good approximation property with an arbitrary seed node, and so we first introduce a seed sampling strategy $\mathbb{P}$ with respect to a set $S^*$ that we wish to find. Informally, the seed selection strategy says that the expected solution mass outside $S^*$ is not too large, and more specifically, not too much larger than if you had seeded on the entire target set $S^*$.  
\begin{definition}\label{def:sampling}
Denote $\vx(\gamma, R)$ as the solution to \eqref{hypergraph-2-norm-cut} with $\kappa=0$. A \emph{good sampling} strategy $\mathbb{P}$ for a target set $S^*$ is
\begin{align*}
 \mathbb{E}_{v\in\mathbb{P}}\left[\frac{1}{d_v}\sum_{u\in V\backslash S^*}d_ux_u(\gamma,\{v\})\right] \leq \frac{c}{\vol(S^*)}\sum_{u\in V\backslash S^*}d_ux_u(\gamma,S^*)
\end{align*}
for some positive constant $c$.
\end{definition}
Note that $\vol(S^*)$ is just to normalize the effect of using different numbers of seeds. For an arbitrary $S^*$, a good sampling strategy $\mathbb{P}$ for the standard graph case with $c=1$ is to sample nodes from $S^*$ proportional to their degree. Now, we provide our main theorem and show its proof in Appendix~\ref{app:cond} .
\begin{theorem}\label{thm:cond}
Given a set $S^*$ of vertices s.t. $\vol(S^*)\leq \frac{\vol(\mathcal{H})}{2}$ and $\phi_{\mathcal{H}}(S^*)\leq \frac{\gamma}{8c}$ for some positive constant $\gamma,c.$ If we have a seed sampling strategy $\mathbb{P}$ that satisfies Def.~\ref{def:sampling}, then with probability at least $\frac{1}{2},$  sweepcut on \eqref{hypergraph-2-norm-cut} with find $S_{\vx}$ with 
\begin{align*}
\phi(S_{\vx}) \leq \sqrt{32\gamma \bar{\delta}  \ln \left(100{\vol(S^*)}/{d_v}\right)},
\end{align*}
where $\bar{\delta}=\max_{e\in \partial S_{\vx}}\min\{\delta_e,|e|/2\}$ where $\partial S_{\vx}=\{e\in\mathcal{E}|e\cap S_{\vx}\neq\emptyset, e\cap \bar{S}_{\vx}\neq \emptyset\}$ and $v$ is the seeded node. 
\end{theorem}

The proof is in the appendix. This implies that for any set $S^*$, if we have a sampling strategy that matches $S^*$ and tune $\gamma$,  our method can find a node set with conductance $O(\sqrt{\phi_{\mathcal{H}}(S^*)\bar{\delta}\log(\vol(S^*))})$. The term $\bar{\delta}$ is the additional cost that we pay for performing graph reduction. The dependence on $\bar{\delta}$ essentially generalizes the previous works that analyzed the conductance with only \emph{all-or-nothing} penalty~\cite{panli2020qdsfm,takai2020hypergraph}, as our result matches these when $\bar{\delta}=1$.  But our method gives the flexibility to choose other values $\delta_e$ and while $\bar{\delta}$ in the worst case could be as large as $|e|/2$, in practice, $\bar{\delta}$ can be chosen much smaller (See \S\ref{sec:exp}). Also, although we reduce $\mathcal{H}$ into a directed graph $G$, the previous conductance analysis for directed graphs~\cite{yoshida2016nonlinear,panli2020qdsfm} is not applicable as we have degree zero nodes in $G$. Those degree zero nodes introduce challenges. 

\section{Directly Related work} We have discussed most related work in-situ throughout the paper. Here, we address a few related hypergraph PageRank vectors directly. First, \citet{panli2020qdsfm} defined a quadratic hypergraph PageRank by directly using Lov\'{a}sz extension of the splitting function $f_e$ to control the diffusion instead of a reduction. Both \citet{panli2020qdsfm} and \citet{takai2020hypergraph} simultaneously proved that this PageRank can be used to partition hypergraphs with an \emph{all-or-nothing} penalty and a Cheeger-type guarantee. Neither approach gives a strongly local algorithm and they have complexity $O(|\mathcal{E}||V|\min\{|\mathcal{E}|,|V|\} \text{poly}(\zeta))$ or in terms of Euler integration or subgradient descent.



\section{Generalization to p-norms}
\label{sec:pnorms}
In the context of the local graph clustering, the quadratic cut objective can sometimes ``over-expand'' or ``bleed out'' over natural boundaries in the data. (This is the opposite problem to the maxflow-based clustering.)  To solve this issue, \cite{liu2020strongly} proposed a more general $p$-norm based cut objective, where $1<p\leq 2$. The corresponding p-norm diffusion algorithm can not only grow from small seed set, but also capture the boundary better than 2-norm cut objective. Moreover, \cite{yang2020p} proposed a related p-norm flow objective that shares similar characteristics. 
Our  \emph{hypergraph diffusion} framework easily adapts to such a generalization.

\mypara{Definition: p-norm local hypergraph diffusions.}
Given a hypergraph $\mathcal{H}=(V,\mathcal{E})$,
seeds $R$, and values $\gamma, \kappa$. Let $\mB, \vw$ again be the incidence matrix and weight vector of the localized reduced directed cut graph.
A $p$-norm local hypergraph diffusion is:
\begin{equation}
\label{hypergraph-p-norm-cut}
\MINone{\vx}{\vw^T \ell((\mB\vx)_+) + \kappa \gamma \sum_{i\in V} x_i d_i}{x_s=1,x_t=0,\vx\geq 0.}
\end{equation}	
Here $\ell(x)=\frac{1}{p}x^p$, $1<p\leq 2$. And the corresponding residual function is $\vr(\vx)=-\frac{1}{\gamma}\mB^T\text{diag}(\ell'((\mB\vx)_+))\vw$.

The idea of solving~\eqref{hypergraph-p-norm-cut} is similar to the quadratic case, where the goal is to iteratively push values to $x_i$ as long as node $i$ violates the optimality condition, i.e. $r_i>\kappa d_i$. The challenge of solving a more general p-norm cut objective is that we no longer have a closed form solution even if the ordering of adjacent nodes is known. Thus, we need to use binary search to find $\Delta x_i$, $\Delta x_a$ and $\Delta x_b$ up to $\eps$ accuracy at every iteration. This means that in the worst case, the general push process can be slower than 2-norm based push process by a factor of $O(\text{log}(1/\eps))$. We defer the details of the algorithm to a longer version of the paper, but we note that a similar analysis shows that this algorithm is strongly local.

\section{Experiments}\label{sec:exp}
In the experiments, we will investigate both the LHQD (2-norm) and 1.4-norm cut objectives with the $\delta$-linear threshold as the splitting function (more details about this function in \S\ref{sec:local-ppr}).  Our focus in this experimental investigation is on the use of the methods for semi-supervised learning. Consequently, we consider how well the algorithms identify ``ground truth'' clusters that represent various known labels in the datasets when given a small set of seeds. (We leave detailed comparisons of the conductances to a longer version.) 

In the plots and tables, we use LH-2.0 to represent our LHQD or LHPR method and LH-1.4 to represent the 1.4 norm version from \S\ref{sec:pnorms}.
The other four methods we compare are:\\
\emph{(i)} ACL~\cite{andersen2006local}, which is initially designed to compute approximated PageRank on graphs. Here we transform each hypergraph to a graph using three different techniques, which are star expansion (star+ACL), unweighted clique expansion (UCE+ACL) and weighted clique expansion (WCE+ACL) where a hyperedge $e$ is replaced by a clique where each edge has weight $1/|e|$~\cite{Zhou2006learning}. ACL is known as one of the fastest and most successful local graph clustering algorithm in several benchmarks~\cite{Veldt-2016-simple-local-flow,liu2020strongly} and has a similar quadratic guarantee on local graph clustering~\cite{andersen2006local,zhu2013local}. 
\\ \emph{(ii)} flow~\cite{veldt2020minimizing}, which is the maxflow-mincut based local method designed for hypergraphs. Since the flow method has difficulty growing from small seed set as illustrated in the yelp experiment in \S\ref{sec:intro}, we will first use the one hop neighborhood to grow the seed set. (OneHop+flow) To limit the number of neighbors included, we will order the neighbors using the \textit{BestNeighbors} as introduced in~\cite{veldt2020minimizing} and only keep at most 1000 neighbors. (Given a seedset $R$, \textit{BestNeighbors} orders nodes based on the fraction of hyperedges incident to $v$ that are also incident to at least one node from $R$.) 
\\ \emph{(iii)} LH-2.0+flow, this is a combination of LH-2.0 and flow where we use the output of LH-2.0 as the input to the flow method to refine. 
\\ \emph{(iv)} HGCRD~\cite{Ibrahim-preprint-hocrd}, this is a hypergraph generalization of CRD~\cite{wang2017capacity}, which is a hybrid diffusion and flow.\footnote{Another highly active topic for clustering and semi-supervised learning involves graph neural networks (GNN). Prior comparisons between GNNs and diffusions shows mixed results in the \emph{small seed set} regime we consider~\cite{Ibrahim-2019-nonlinear,liu2020strongly} and complicates doing a fair comparison. As such, we focus on comparing with the most directly related work.}

In order to select an appropriate $\delta$ for different datasets, Veldt et al.~found that the optimal $\delta$ is usually consistent among different clusters in the same dataset~\cite{veldt2020minimizing}. Thus, the optimal $\delta$ can be visually approximated by varying $\delta$ for a handful of clusters if one has access to a subset of ground truth clusters in a hypergraph. We adapt the same procedure in our experiments and report the results in App.~\ref{sec:select-delta}.
Other parameters are in the reproduction details footnote.\footnote{\emph{Reproduction details}. The full algorithm and evaluation codes can be found here \url{https://github.com/MengLiuPurdue/LHQD}. We fix the LH locality parameter $\gamma$ to be 0.1, approximation parameter $\rho$ to be 0.5 in all experiments. We set 
$\kappa=0.00025$ for Amazon and $\kappa=0.0025$ for Stack Overflow based on cluster size. For ACL, we use the same set of parameters as LH. For LH-2.0+flow, we set the flow method's locality parameter to be 0.1. For OneHop+flow, we set the locality parameter to be 0.05, 0.0025 on Amazon and Stack Overflow accordingly. For HGCRD, we set $U=3$ (maximum flow that can be send out of a node), $h=3$ (maximum flow that an edge can handle), $w=2$ (multiplicative factor for increasing the capacity of the nodes at each iteration), $\alpha=1$ (controls the eligibility of hyperedge), $\tau=0.5$ and 6 maximum iterations.}


\subsection{Detecting Amazon Product Categories}
In this experiment, we use different methods to detect Amazon product categories~\cite{ni-etal-2019-justifying}. The hypergraph is constructed from Amazon product review data where each node represents a product and each hyperedge is set of products reviewed by the same person. It has 2,268,264 nodes and 4,285,363 hyperedges. The average size of hyperedges is around 17. We select 6 different categories with size between 100 and 10000 as ground truth clusters used in~\cite{veldt2020minimizing}. 
We set $\delta=1$ for this dataset (more details about this choice in \S\ref{sec:select-delta}). We select 1\% nodes (at least 5) as seed set for each cluster and report median F1 scores and median runtime over 30 trials in Table~\ref{tab:amazon_f1} and~\ref{tab:amazon_runtime}. Overall, LH-1.4 has the best F1 scores and LH-2.0 has the second best F1. 
The two fastest methods are LH-2.0 and star+ACL. 
While achieving better F1 scores, LH-2.0 is 20x faster than HyperLocal (flow) and 2-5x faster than clique expansion based methods. 

\begin{table}[t]
  \caption{Median F1 scores on detecting Amazon product categories over 30 trials, the small violin plots show variance. }
  \confspace{-\baselineskip}
  \label{tab:amazon_f1}
  \noindent \begin{tabularx}{\linewidth}{@{}l@{\,\,}*{6}{@{}l@{}X@{}}@{}}
\toprule
 Alg & 
\mbox{\rlap{12}} & &
\mbox{\rlap{18}} & &
\mbox{\rlap{17}} & &
\mbox{\rlap{25}} & &
\mbox{\rlap{15}} & &
\mbox{\rlap{24}} &  \\
 & \mbox{\rlap{\footnotesize F1 \& Med.}}  & &
\mbox{\rlap{\footnotesize F1 \& Med.}} & &
\mbox{\rlap{\footnotesize F1 \& Med.}} & &
\mbox{\rlap{\footnotesize F1 \& Med.}} & &
\mbox{\rlap{\footnotesize F1 \& Med.}} & &
\mbox{\rlap{\footnotesize F1 \& Med.}} &    \\
\midrule
 LH-2.0 & \mbox{\rlap{\raisebox{-2pt}{\includegraphics[width=20pt,height=10pt]{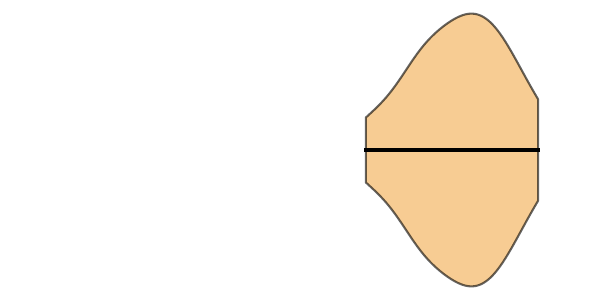}}}}\mbox{\rlap{\mbox{\hspace{12pt}0.77}}} & & \mbox{\rlap{\raisebox{-2pt}{\includegraphics[width=20pt,height=10pt]{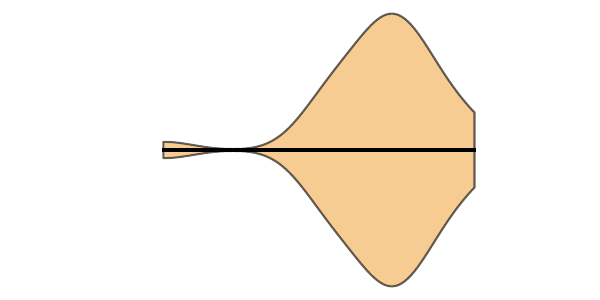}}}}\mbox{\rlap{\mbox{\hspace{10pt}0.65}}} & & \mbox{\rlap{\raisebox{-2pt}{\includegraphics[width=20pt,height=10pt]{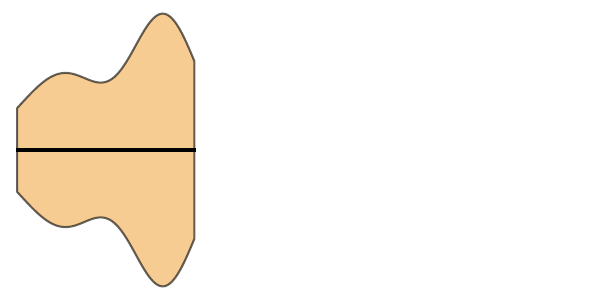}}}}\mbox{\rlap{\mbox{\hspace{4pt}0.25}}} & & \mbox{\rlap{\raisebox{-2pt}{\includegraphics[width=20pt,height=10pt]{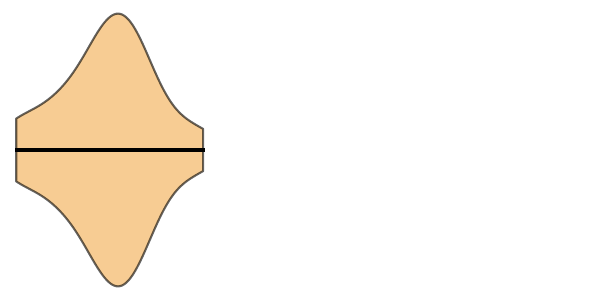}}}}\mbox{\rlap{\mbox{\hspace{3pt}0.19}}} & & \mbox{\rlap{\raisebox{-2pt}{\includegraphics[width=20pt,height=10pt]{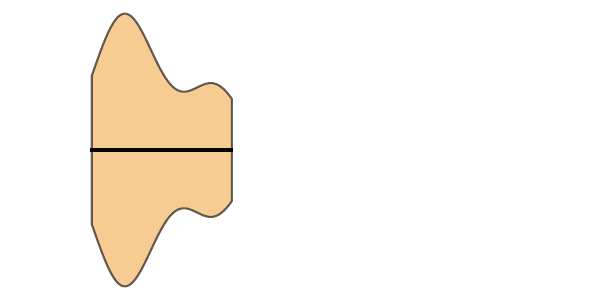}}}}\mbox{\rlap{\mbox{\hspace{3pt}0.22}}} & & \mbox{\rlap{\raisebox{-2pt}{\includegraphics[width=20pt,height=10pt]{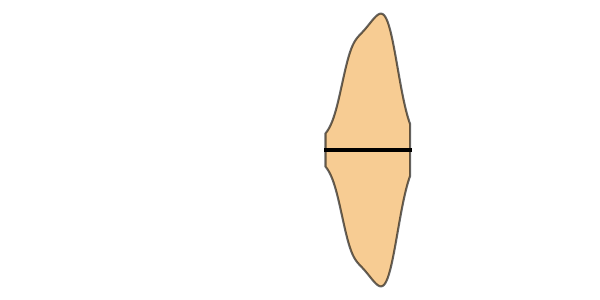}}}}\mbox{\rlap{\mbox{\hspace{9pt}0.62}}} & \\ 
 LH-1.4 & \mbox{\rlap{\raisebox{-2pt}{\includegraphics[width=20pt,height=10pt]{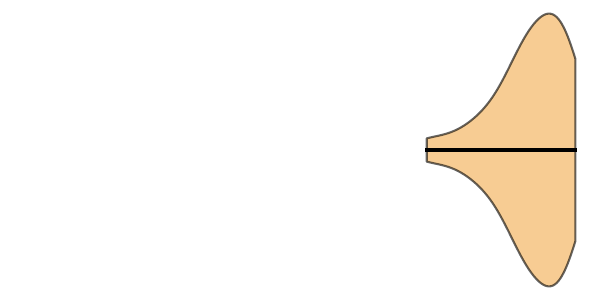}}}}\mbox{\rlap{\mbox{\hspace{14pt}0.9}}} & & \mbox{\rlap{\raisebox{-2pt}{\includegraphics[width=20pt,height=10pt]{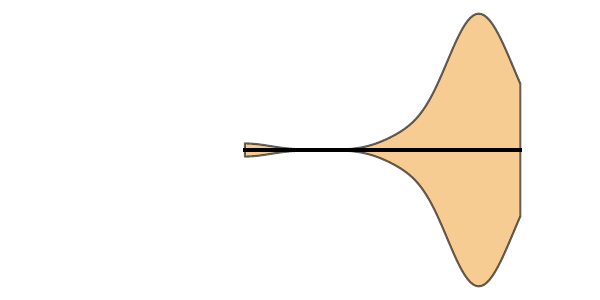}}}}\mbox{\rlap{\mbox{\hspace{12pt}0.79}}} & & \mbox{\rlap{\raisebox{-2pt}{\includegraphics[width=20pt,height=10pt]{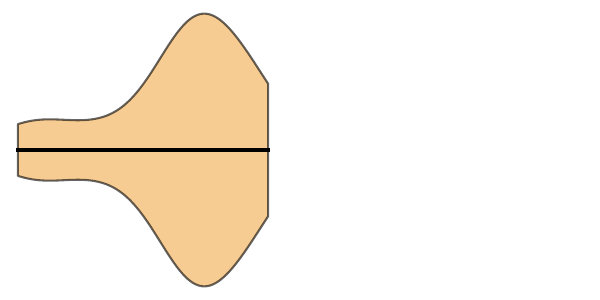}}}}\mbox{\rlap{\mbox{\hspace{5pt}0.32}}} & & \mbox{\rlap{\raisebox{-2pt}{\includegraphics[width=20pt,height=10pt]{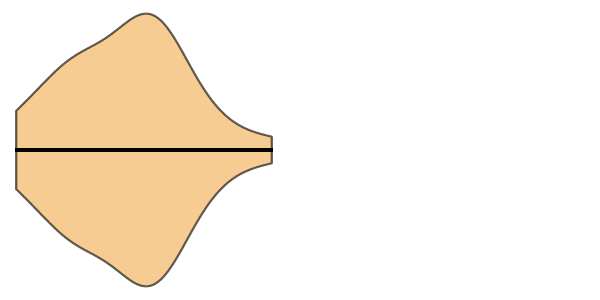}}}}\mbox{\rlap{\mbox{\hspace{3pt}0.22}}} & & \mbox{\rlap{\raisebox{-2pt}{\includegraphics[width=20pt,height=10pt]{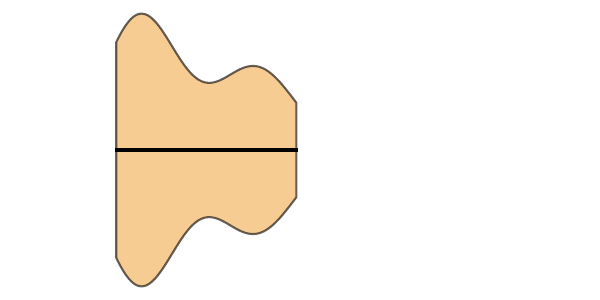}}}}\mbox{\rlap{\mbox{\hspace{4pt}0.27}}} & & \mbox{\rlap{\raisebox{-2pt}{\includegraphics[width=20pt,height=10pt]{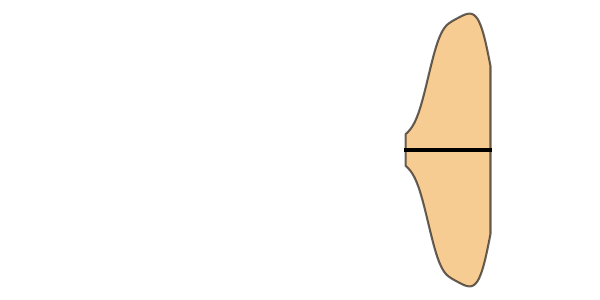}}}}\mbox{\rlap{\mbox{\hspace{12pt}0.77}}} & \\ 
 LH-2.0+flow & \mbox{\rlap{\raisebox{-2pt}{\includegraphics[width=20pt,height=10pt]{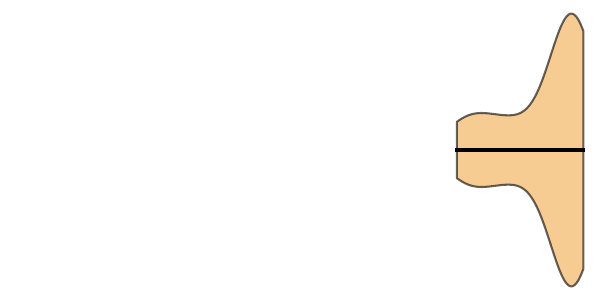}}}}\mbox{\rlap{\mbox{\hspace{14pt}0.95}}} & & \mbox{\rlap{\raisebox{-2pt}{\includegraphics[width=20pt,height=10pt]{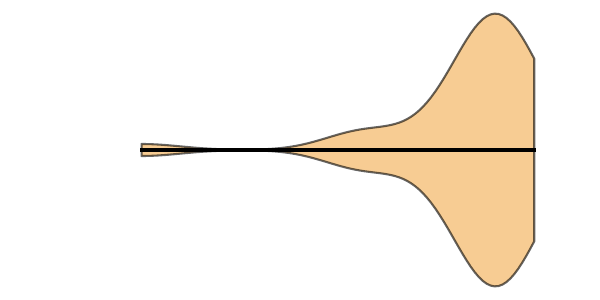}}}}\mbox{\rlap{\mbox{\hspace{12pt}0.82}}} & & \mbox{\rlap{\raisebox{-2pt}{\includegraphics[width=20pt,height=10pt]{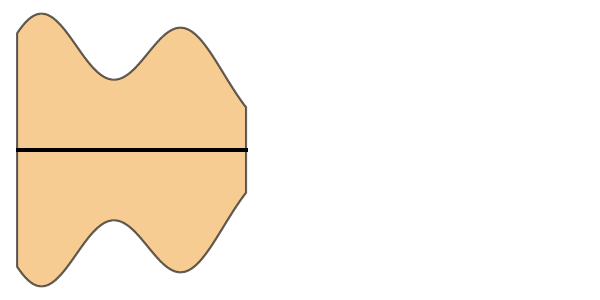}}}}\mbox{\rlap{\mbox{\hspace{2pt}0.15}}} & & \mbox{\rlap{\raisebox{-2pt}{\includegraphics[width=20pt,height=10pt]{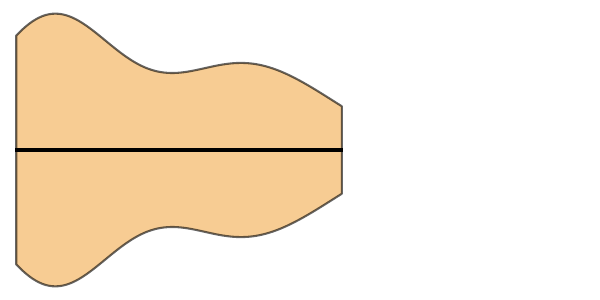}}}}\mbox{\rlap{\mbox{\hspace{2pt}0.16}}} & & \mbox{\rlap{\raisebox{-2pt}{\includegraphics[width=20pt,height=10pt]{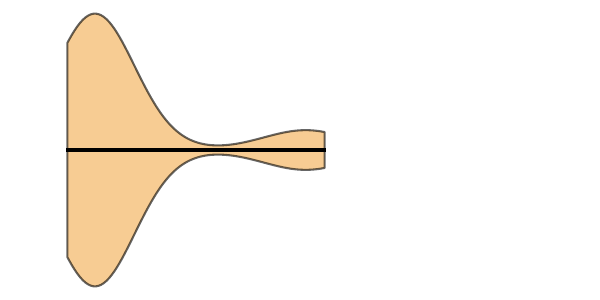}}}}\mbox{\rlap{\mbox{\hspace{2pt}0.16}}} & & \mbox{\rlap{\raisebox{-2pt}{\includegraphics[width=20pt,height=10pt]{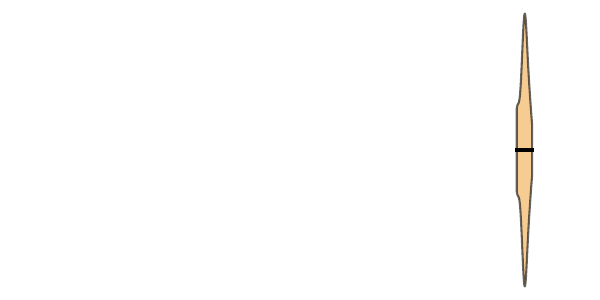}}}}\mbox{\rlap{\mbox{\hspace{13pt}0.87}}} & \\ 
 star+ACL & \mbox{\rlap{\raisebox{-2pt}{\includegraphics[width=20pt,height=10pt]{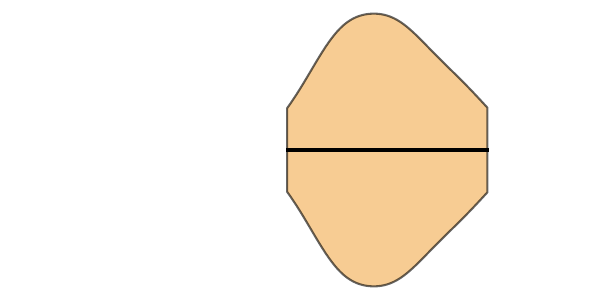}}}}\mbox{\rlap{\mbox{\hspace{10pt}0.64}}} & & \mbox{\rlap{\raisebox{-2pt}{\includegraphics[width=20pt,height=10pt]{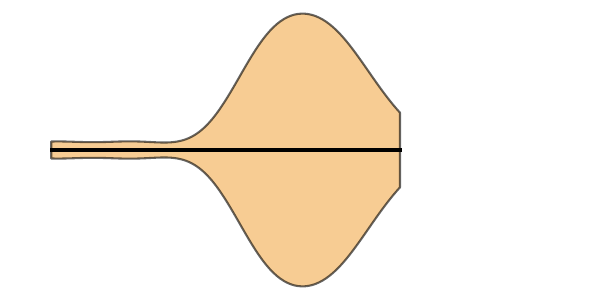}}}}\mbox{\rlap{\mbox{\hspace{8pt}0.51}}} & & \mbox{\rlap{\raisebox{-2pt}{\includegraphics[width=20pt,height=10pt]{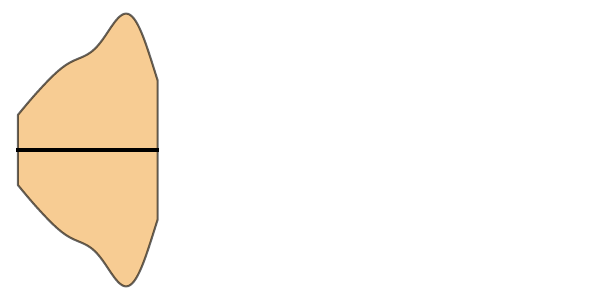}}}}\mbox{\rlap{\mbox{\hspace{3pt}0.19}}} & & \mbox{\rlap{\raisebox{-2pt}{\includegraphics[width=20pt,height=10pt]{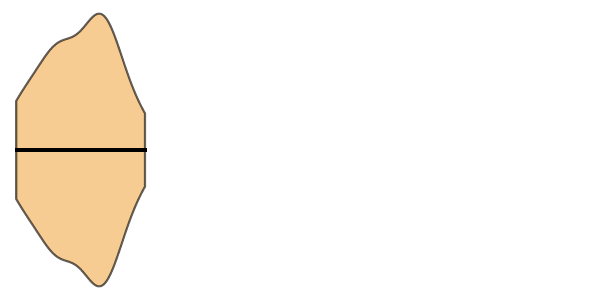}}}}\mbox{\rlap{\mbox{\hspace{2pt}0.15}}} & & \mbox{\rlap{\raisebox{-2pt}{\includegraphics[width=20pt,height=10pt]{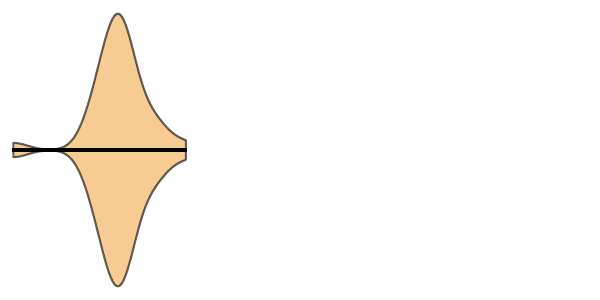}}}}\mbox{\rlap{\mbox{\hspace{3pt}0.2}}} & & \mbox{\rlap{\raisebox{-2pt}{\includegraphics[width=20pt,height=10pt]{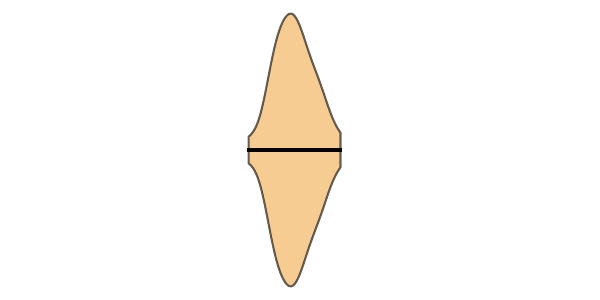}}}}\mbox{\rlap{\mbox{\hspace{7pt}0.49}}} & \\ 
 WCE+ACL & \mbox{\rlap{\raisebox{-2pt}{\includegraphics[width=20pt,height=10pt]{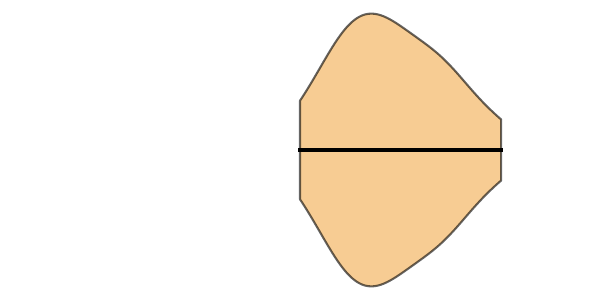}}}}\mbox{\rlap{\mbox{\hspace{10pt}0.64}}} & & \mbox{\rlap{\raisebox{-2pt}{\includegraphics[width=20pt,height=10pt]{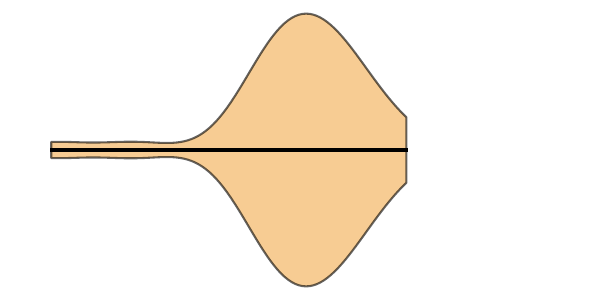}}}}\mbox{\rlap{\mbox{\hspace{8pt}0.51}}} & & \mbox{\rlap{\raisebox{-2pt}{\includegraphics[width=20pt,height=10pt]{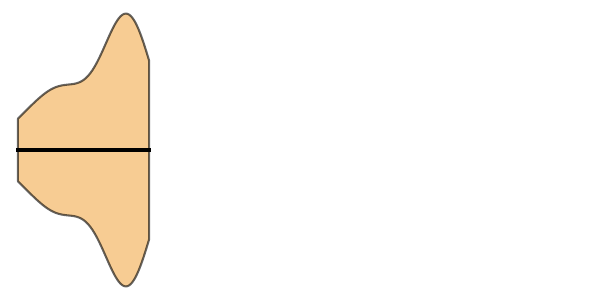}}}}\mbox{\rlap{\mbox{\hspace{3pt}0.2}}} & & \mbox{\rlap{\raisebox{-2pt}{\includegraphics[width=20pt,height=10pt]{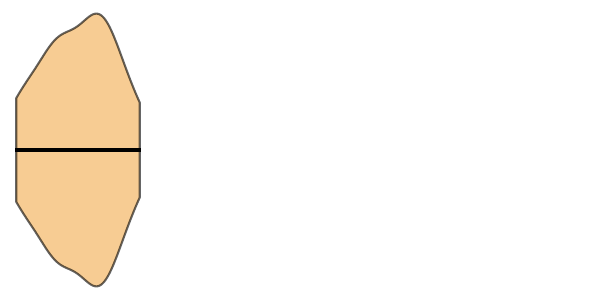}}}}\mbox{\rlap{\mbox{\hspace{2pt}0.14}}} & & \mbox{\rlap{\raisebox{-2pt}{\includegraphics[width=20pt,height=10pt]{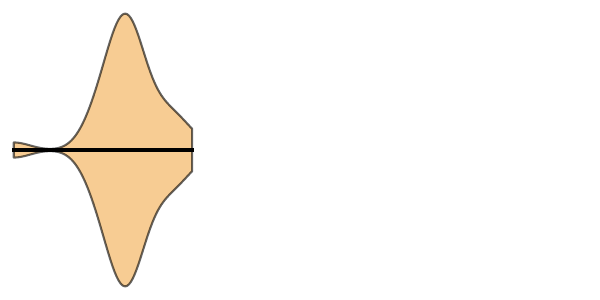}}}}\mbox{\rlap{\mbox{\hspace{3pt}0.21}}} & & \mbox{\rlap{\raisebox{-2pt}{\includegraphics[width=20pt,height=10pt]{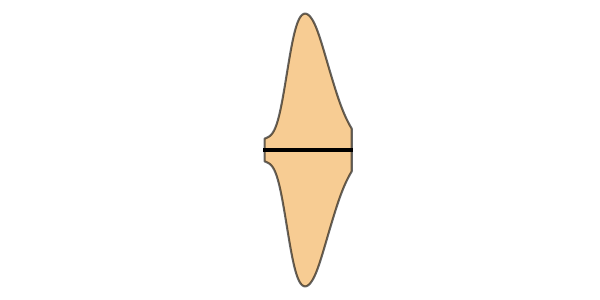}}}}\mbox{\rlap{\mbox{\hspace{8pt}0.51}}} & \\ 
 UCE+ACL & \mbox{\rlap{\raisebox{-2pt}{\includegraphics[width=20pt,height=10pt]{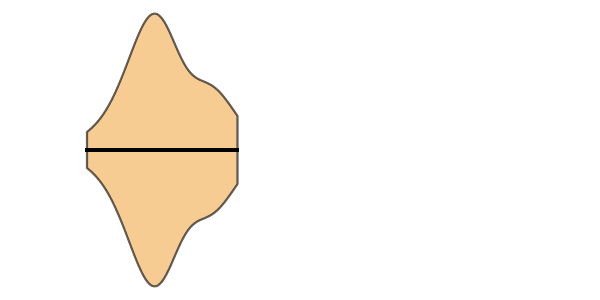}}}}\mbox{\rlap{\mbox{\hspace{4pt}0.27}}} & & \mbox{\rlap{\raisebox{-2pt}{\includegraphics[width=20pt,height=10pt]{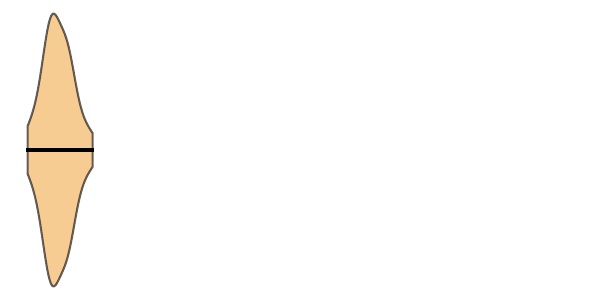}}}}\mbox{\rlap{\mbox{\hspace{1pt}0.09}}} & & \mbox{\rlap{\raisebox{-2pt}{\includegraphics[width=20pt,height=10pt]{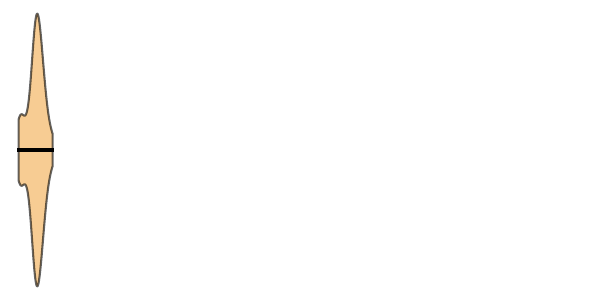}}}}\mbox{\rlap{\mbox{\hspace{1pt}0.06}}} & & \mbox{\rlap{\raisebox{-2pt}{\includegraphics[width=20pt,height=10pt]{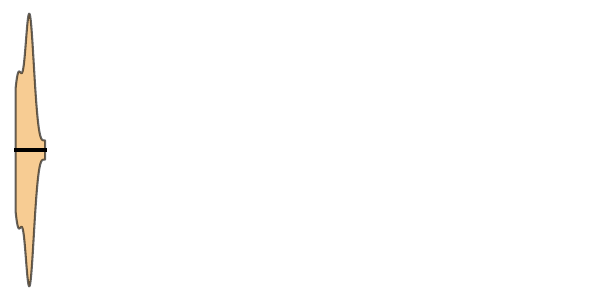}}}}\mbox{\rlap{\mbox{\hspace{1pt}0.05}}} & & \mbox{\rlap{\raisebox{-2pt}{\includegraphics[width=20pt,height=10pt]{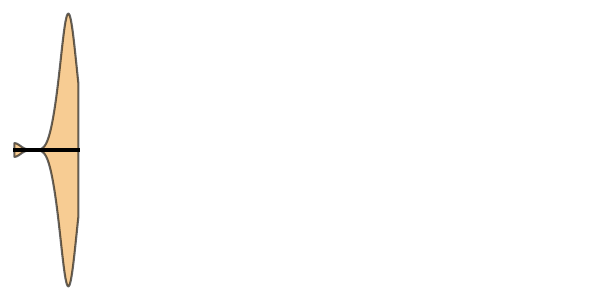}}}}\mbox{\rlap{\mbox{\hspace{2pt}0.11}}} & & \mbox{\rlap{\raisebox{-2pt}{\includegraphics[width=20pt,height=10pt]{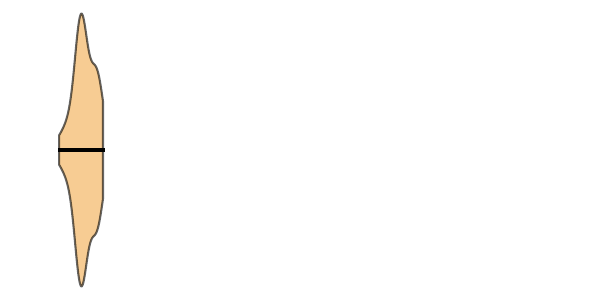}}}}\mbox{\rlap{\mbox{\hspace{2pt}0.14}}} & \\ 
 OneHop+flow & \mbox{\rlap{\raisebox{-2pt}{\includegraphics[width=20pt,height=10pt]{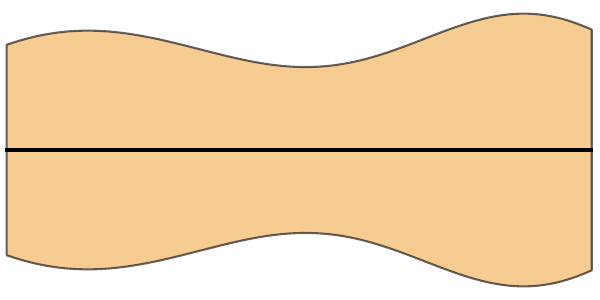}}}}\mbox{\rlap{\mbox{\hspace{8pt}0.52}}} & & \mbox{\rlap{\raisebox{-2pt}{\includegraphics[width=20pt,height=10pt]{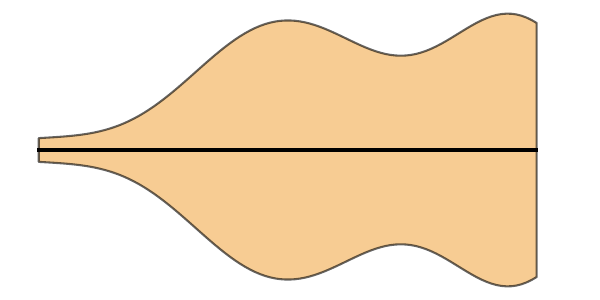}}}}\mbox{\rlap{\mbox{\hspace{9pt}0.6}}} & & \mbox{\rlap{\raisebox{-2pt}{\includegraphics[width=20pt,height=10pt]{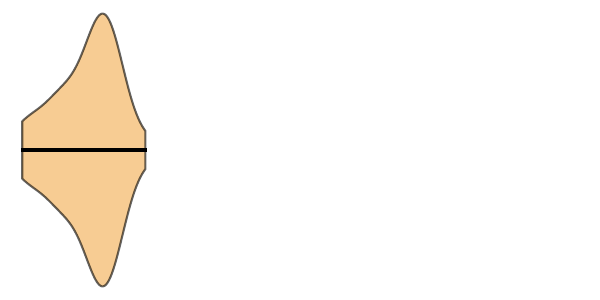}}}}\mbox{\rlap{\mbox{\hspace{2pt}0.16}}} & & \mbox{\rlap{\raisebox{-2pt}{\includegraphics[width=20pt,height=10pt]{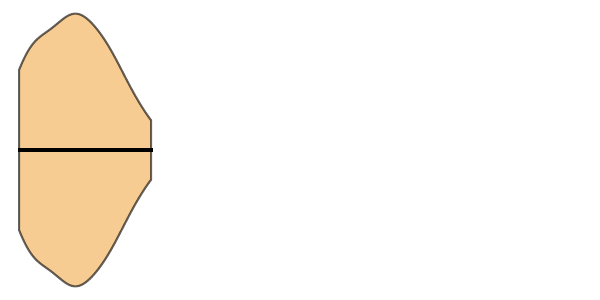}}}}\mbox{\rlap{\mbox{\hspace{2pt}0.12}}} & & \mbox{\rlap{\raisebox{-2pt}{\includegraphics[width=20pt,height=10pt]{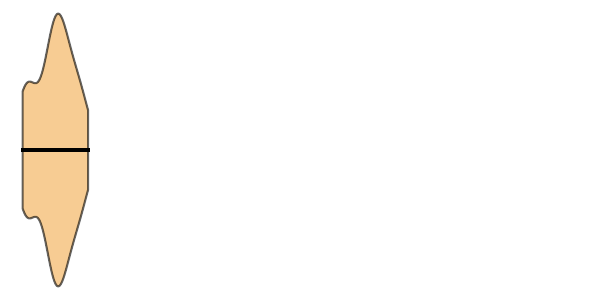}}}}\mbox{\rlap{\mbox{\hspace{1pt}0.09}}} & & \mbox{\rlap{\raisebox{-2pt}{\includegraphics[width=20pt,height=10pt]{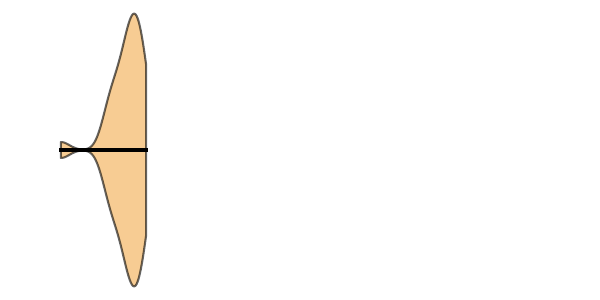}}}}\mbox{\rlap{\mbox{\hspace{3pt}0.22}}} & \\ 
 HGCRD & \mbox{\rlap{\raisebox{-2pt}{\includegraphics[width=20pt,height=10pt]{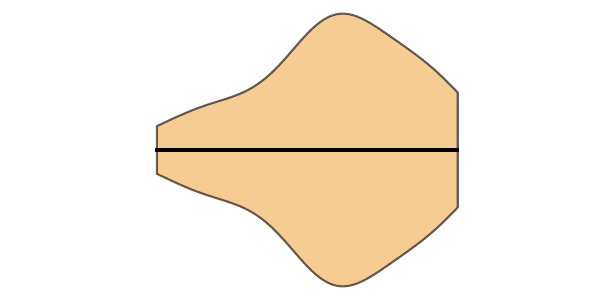}}}}\mbox{\rlap{\mbox{\hspace{2pt}0.56}}} & & \mbox{\rlap{\raisebox{-2pt}{\includegraphics[width=20pt,height=10pt]{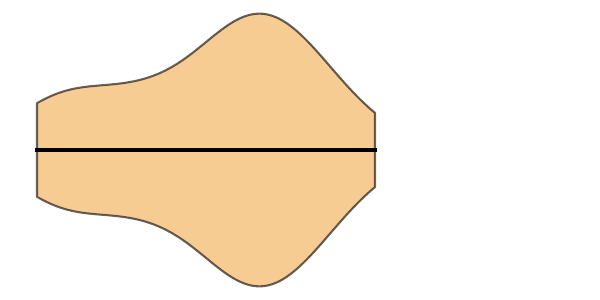}}}}\mbox{\rlap{\mbox{\hspace{2pt}0.4}}} & & \mbox{\rlap{\raisebox{-2pt}{\includegraphics[width=20pt,height=10pt]{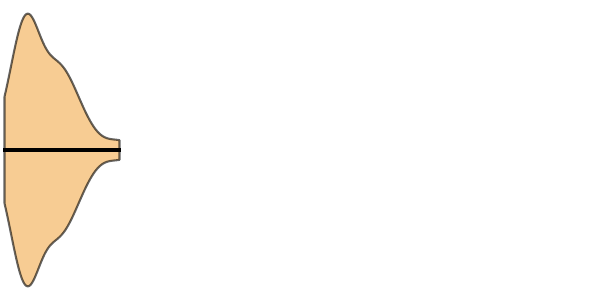}}}}\mbox{\rlap{\mbox{\hspace{1pt}0.05}}} & & \mbox{\rlap{\raisebox{-2pt}{\includegraphics[width=20pt,height=10pt]{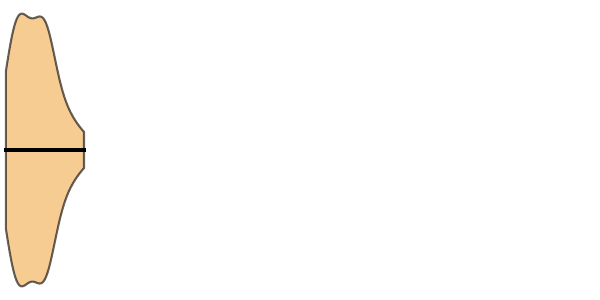}}}}\mbox{\rlap{\mbox{\hspace{1pt}0.06}}} & & \mbox{\rlap{\raisebox{-2pt}{\includegraphics[width=20pt,height=10pt]{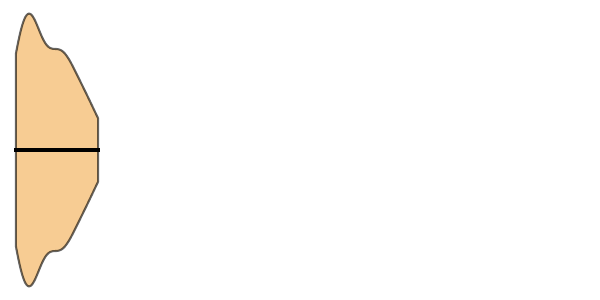}}}}\mbox{\rlap{\mbox{\hspace{1pt}0.07}}} & & \mbox{\rlap{\raisebox{-2pt}{\includegraphics[width=20pt,height=10pt]{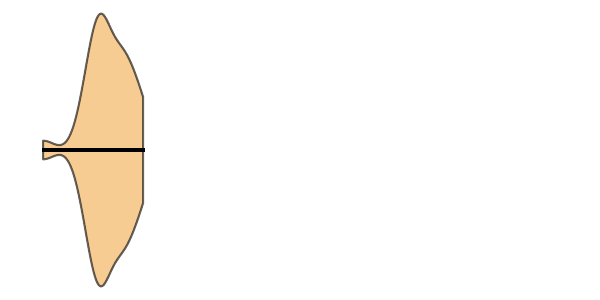}}}}\mbox{\rlap{\mbox{\hspace{3pt}0.17}}} & \\ 
\bottomrule 
\end{tabularx}

  \confspace{-\baselineskip}
\end{table}

\begin{table}[t]
  \caption{Median runtime in seconds on detecting Amazon product categories}
  \confspace{-\baselineskip}
  \label{tab:amazon_runtime}
  \noindent \begin{tabularx}{\linewidth}{@{}lXXXXXX@{}}
\toprule
 Alg & 12 &18 &17 &25 &15 &24 \\
\midrule
LH-2.0 &0.9 &0.7 &2.8 &1.0 &5.6 &13.3\\ 
LH-1.4 &8.0 &6.3 &32.3 &9.8 &53.8 &127.3\\ 
LH-2.0+flow &3.5 &5.1 &421.1 &17.8 &34.9 &151.5\\ 
star+ACL &0.2 &0.2 &0.3 &0.2 &0.5 &0.8\\ 
WCE+ACL &18.6 &17.2 &19.0 &16.5 &21.5 &20.1\\ 
UCE+ACL &9.8 &10.9 &11.2 &10.7 &13.3 &15.5\\ 
OneHop+flow &308.8 &141.7 &359.2 &224.9 &81.5 &82.4\\ 
HGCRD &120.3 &56.4 &78.1 &21.2 &239.4 &541.3\\ 
\bottomrule 
\end{tabularx}

\end{table}

\subsection{Detecting Stack Overflow Question Topics}
In the Stack Overflow dataset, we have a hypergraph with each node representing a question on ``stackoverflow.com'' and each hyperedge representing questions answered by the same user~\cite{veldt2020minimizing}. Each question is associated with a set of tags. The goal is to find questions having the same tag when seeding on some nodes with a given target tag.
 This hypergraph is much larger with 15,211,989 nodes and 1,103,243 edges. The average hyperedge size is around 24. We select 40 clusters with 2,000 to 10,000 nodes and a conductance score below 0.2 using the \emph{all-or-nothing} penalty. (There are 45 clusters satisfying these conditions, 5 of them are used to select $\delta$.)
In this dataset, a large $\delta$ can give better results. For diffusion based methods, we set the $\delta$-linear threshold to be $1000$ (more details about this choice in App.~\S\ref{sec:select-delta}), while for flow based method, we set $\delta=5000$ based on Figure 3 of~\cite{veldt2020minimizing}. In Table~\ref{tab:stackoverflow}, we summarize some recovery statistics of different methods on this dataset. In Figure~\ref{stackoverflow-f1-runtime}, we report the performance of different methods on each cluster. Overall, LH-2.0 achieves the best balance between speed and accuracy, although all the diffusion based methods (LH, ACL) have extremely similar F1 scores (which is different from the last experiment). The flow based method still has difficulty growing from small seed set as we can see from the low recall in Table~\ref{tab:stackoverflow}.

\begin{figure}[t]
  \centering
  \includegraphics[width=0.8\linewidth]{"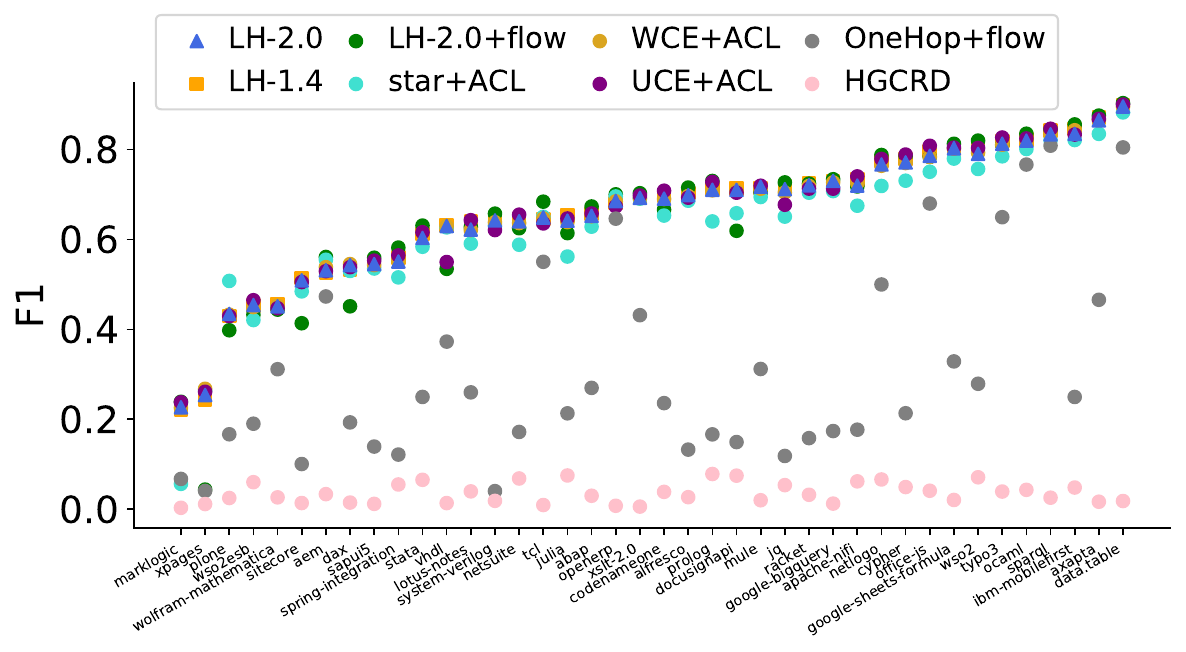"}
  \includegraphics[width=0.8\linewidth]{"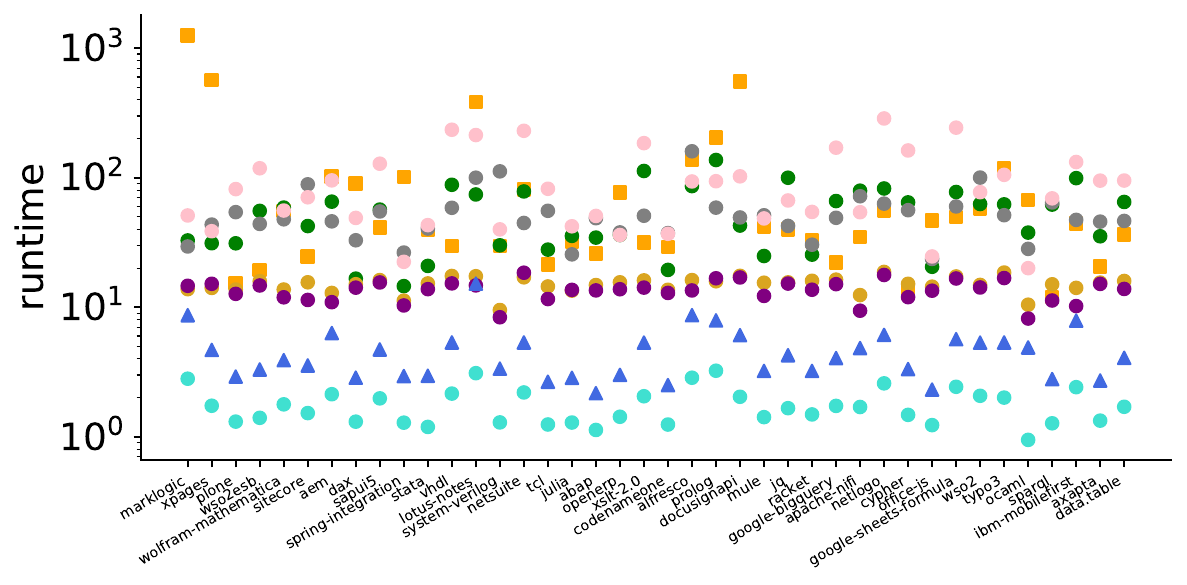"}
  \confspace{-0.5\baselineskip}
  \caption{The upper plot shows median F1 scores of different methods over 40 clusters from the Stack Overflow dataset.  The lower plot shows median running time. LH-2.0 achieves the best balance between speed and accuracy; LH-1.4 can sometimes be slower than the flow method when the target cluster contains many large hyperedges.}
  \confspace{-\baselineskip}
  \label{stackoverflow-f1-runtime}
\end{figure}

\begin{table}[t]
  \caption{This table summarizes the median of median runtimes in seconds for the Stack Overflow experiments as well as median Precision, Recall and F1 over the 40 clusters.}
  \confspace{-\baselineskip}
  \label{tab:stackoverflow}
  \begin{tabularx}{\linewidth}{@{}lXXXXXXXX@{}}
\toprule
Alg. & {LH2} & {LH1.4} & {LH2} & ACL &ACL& ACL & Flow & HG- \\
\cmidrule(l){2-4}
\cmidrule(l){5-7}
\cmidrule(l){8-8}
& & & +flow 
& +star & +WCE & +UCE 
& +1Hop 
& CRD
\\
\midrule
Time & 3.69 & 39.89 & 43.84 & 1.54 & 15.25 & 13.71 & 48.28 & 72.31\\
Pr & 0.65 & 0.66 & 0.74 & 0.66 & 0.65 & 0.66 & 0.83 & 0.46\\
Rc & 0.67 & 0.67 & 0.59 & 0.6 & 0.66 & 0.65 & 0.11 & 0.01\\
F1 & 0.66 & 0.66 & 0.66 & 0.63 & 0.65 & 0.65 & 0.19 & 0.02\\
\bottomrule
\end{tabularx}

  \confspace{-\baselineskip}
\end{table}




\subsection{Varying Number of Seeds}
Comparing to the flow-based method, an advantage of our diffusion method is that it expands from a small seed set into a good enough cluster that detects many labels correctly. Now, we use the Amazon dataset to elaborate on this point. We vary the ratio of seed set from 0.1\% to 10\%. At each seed ratio, denoted as $r$, we set $\kappa=0.025r$. And for each of the 6 clusters, we take the median F1 score over 10 trials. To have a global idea of how different methods perform on this dataset, we take another median over the 6 median F1 scores. For the flow-based method, we also consider removing the OneHop growing procedure. The results are summarized in Figure~\ref{vary-seeds}. We can see our hypergraph diffusion based method (LH-1.4, LH-2.0) performs better than alternatives for all seed sizes, although flow dramatically improves for large seed sizes.
\begin{figure}[t]
  \centering
  \includegraphics[width=0.8\linewidth]{"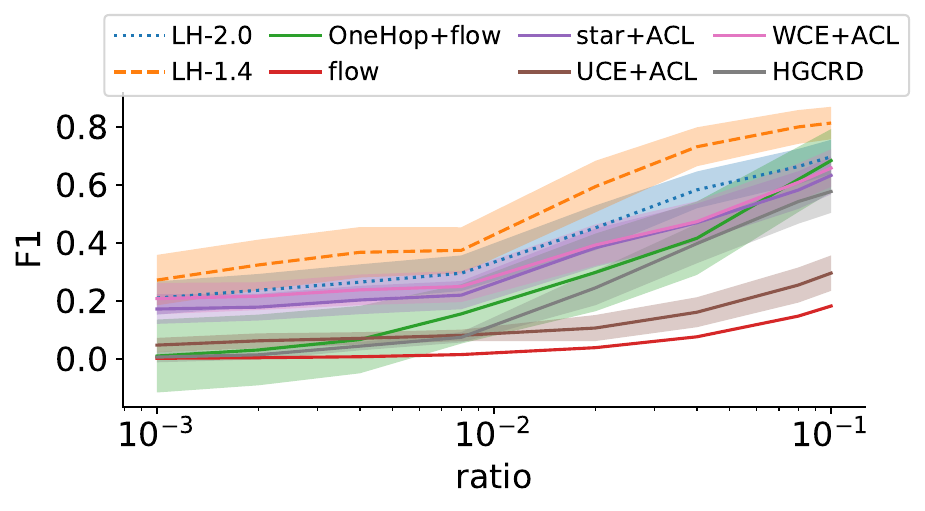"}
  \confspace{-\baselineskip}
  \caption{This plot shows the median of median F1 scores on detecting those 6 clusters in the Amazon data when varying the seed size. The envelope represents 1 standard error over the 6 median F1 scores. Without OneHop, the flow based method is not able to grow from seed set even for the largest seeds. Our hypergraph diffusion (LH) methods outperforms others, especially for small seeds.}
  \label{vary-seeds}
  \confspace{-\baselineskip}
\end{figure}



    
\section{DISCUSSION}

This paper studies the opportunities for strongly local quadratic and $p$-norm diffusions in the context of local clustering and semi-supervised learning. 

One of the distinct challenges we encountered in preparing this manuscript was comparing against the ideas of others. 
Clique expansions are often problematic because they can involve quadratic memory for each hyperedge if used simplistically. For running the baseline ACL PageRank diffusion on the clique expansion, we were able to use the \emph{linear} nature of this algorithm to implicitly model the clique expansion without realizing the actual graph in memory. (We lack space to describe this though.) For others the results were less encouraging.
We, for instance, were unable to set integration parameters for the Euler scheme employed by QHPR~\cite{takai2020hypergraph} to produce meaningful results in \S\ref{sec:exp}. 

Consequently, we wish to discuss the parameters of our method and why they are reasonably easy to set. The values $\gamma$ and $\kappa$ both control the size of the result. Roughly, $\gamma$ corresponds to how much the diffusion is allowed to spread from the seed vertices and $\kappa$ controls how aggressively we \emph{sparsify} the diffusion. To get a bigger result, then, set $\gamma$ or $\kappa$ a little bit smaller. The value of $\rho$ corresponds only to how much one of our solutions can differ from the \emph{unique} solution of~\eqref{hypergraph-2-norm-cut}. Fixing $\rho=0.5$ is fine for empirical studies unless the goal is to compare against other strategies to solve that same equation. 
The final parameter is $\delta$, which interpolates between the \emph{all-or-nothing} penalty and the \emph{cardinality} penalty as discussed in \S\ref{sec:lhpr}. This can be chosen based on an experiment as we did here, or by exploring a few small choices between $1$ and half the largest hyperedge size. 

In closing, \emph{flow-based} algorithms have often been explained or used as \emph{refinement} operations to the clusters produced by spectral methods~\cite{Lang2005-spectral-weaknesses} as in LH-2.0+flow. As a final demonstration of this usage on the Yelp experiment from \S\ref{sec:intro}, we have precision, recall, and F1 result of $0.87, 0.998, 0.93$, respectively, which demonstrates how these techniques may be easily combined to even more accurately find target clusters.

\confspace{-\baselineskip}
\bibliographystyle{ACM-Reference-Format}
\bibliography{sample-base}

\appendix

\section{Proof of Theorem~\ref{runtime-guarantee}}
First we have the following observations on $\vr$ and $\vx$.
\begin{lemma}
\label{aux-lemma}
At any iteration of Algorithm~\ref{cut-algo}, for each pair of auxiliary nodes, $a\in V_a$, $b\in V_b$ and $a\rightarrow b$, $x_a\geq x_b$.
\end{lemma}
\begin{proof}
If $x_a<x_b$, from equation~\eqref{ra-rb}, $r_b=0$ means $x_b\leq x_i$ for any $b\rightarrow i$. But this means $r_a>0$, which is a contradiction.
\end{proof}

\begin{lemma}
\label{positivity}
At any iteration of Algorithm~\ref{cut-algo}, for any $i\in V\cup V_a\cup V_b$, $g_i$ will stay nonnegative and $0\leq x_i\leq 1$.
\end{lemma}
\begin{proof}
The nonnegativity of $g_i$ is obvious. At any iteration of Algorithm~\ref{cut-algo}, we either change $g_i$ to $\rho\kappa d_i$ or add some nonnegative values to the residual of adjacent nodes or keep $g_i$ as zero. For each original node $i\in V$, to prove $0\leq x_i\leq 1$, recall $r_i = g_i$ expands to~\eqref{ri}.
Suppose $x_i$ is the largest element and $x_i>1$, then $\text{Ind}(i\in R)-x_i<0$. Since $g_i\geq 0$, this means there exists $b\in V_b$ such that $x_b>x_i>1$. Assume $a\in V_a$ is adjacent to $b$, then from equation~\eqref{ra-rb}, $g_a=g_b=0$.
 Since $x_i$ is the largest among all nodes $i\in V$, the only way to satisfy both equations is $x_a=x_b=x_i$. Thus, we have $x_i=x_b>x_i$, which is a contradiction. To prove $x_i\leq 1$ for $i\in V_a\cup V_b$, from Lemma~\ref{aux-lemma}, we only need to show $x_i\leq 1$ for $i\in V_a$. If there exists $a\in V_a$, $x_a>1$, since $x_i\leq 1$ for any $i\in V$, then $g_a=0$ means $x_a=x_b>1$ where $a\rightarrow b$. This means $g_b<0$, which is a contradiction.
\end{proof}

\mypara{Proof of Theorem~\ref{runtime-guarantee}.} 
By using Lemma~\ref{positivity}, $\|\vg\|_1$ becomes
\[\textstyle \|\vg\|_1=\sum_{i\in V\cup V_a\cup V_b}g_i=\sum_{i\in R}d_i(1-x_i)-\sum_{i\in\bar{R}}d_ix_i\]
This implies that any change to the auxiliary nodes will not affect $||\vg||_1$. Thus calling \texttt{LHQD-auxpush} doesn't change $||\vg||_1$. When there exists $i\in V$ such that $g_i>\kappa d_i$, then \texttt{hyper-push} will find $\Delta x_i$ such that $g'_i=\rho\kappa d_i$. Then the new $g'_i$ can be written as
\begin{multline}
\label{updated_residual}
g_i'=\frac{1}{\gamma}\sum_{b\in V_b}w_{bi}(x_b-x_i-\Delta x_i)_+-\frac{1}{\gamma}\sum_{a\in V_a}w_{ia}(x_i+\Delta x_i-x_a)_++\\
\kappa d_i(\text{Ind}[i\in R]-x_i-\Delta x_i)
\end{multline}
Note $g_i'$ is a decreasing function of $\Delta x_i$ and $g_i'>0$ when $\Delta x_i=0$, $g_i'<0$ when $\Delta x_i=1$ by using Lemma~\ref{positivity}. This suggests that there exists a unique $\Delta x_i$ that satisfies $g_i'=\rho\kappa d_i$. Moreover, $(x_b-x_i-\Delta x_i)_+\geq (x_b-x_i)_+-\Delta x_i$ and $(x_i+\Delta x_i-x_a)_+\leq (x_i-x_a)_++\Delta x_i$, thus we have
\[\textstyle \rho\kappa d_i=g_i'\geq g_i-\frac{1}{\gamma}(\sum_{b\in V_b}w_{bi}+\sum_{a\in V_a}w_{ia})\Delta x_i-\kappa d_i\Delta x_i\] 
From equation (4.9) of~\cite{veldt2020hypercuts},
\[\textstyle  \sum_{b\in V_b}w_{bi}=\sum_{a\in V_a}w_{ai}=\sum_{e\in\mathcal{E},i\in e}f_e(\{i\})\leq\delta d_i\]
Thus, we have $\textstyle  d_i\Delta x_i\geq \frac{g_i-g_i'}{\kappa+\delta/\gamma}>\frac{\gamma\kappa(1-\rho)}{\gamma\kappa+\delta}d_i$
So the decrease of $||\vg||_1$ will be at least $\gamma\kappa(1-\rho)d_i/(\gamma\kappa+\delta)$. Since $||\vg||_1=\text{vol}(R)$ initially, we have
$\textstyle \sum_{i=1}^Td_i\leq(\gamma\kappa+\delta)\text{vol}(R)/\gamma\kappa(1-\rho)=O(\text{vol}(R)) \qedhere$.





\section{Proof of Theorem~\ref{thm:cond}} \label{app:cond} 

We first introduce some simplifying notation. We use $1_v$ to denote the canonical basis with $v$th component as 1 and others as 0 and $1_S = \sum_{v\in S}1_v$. Without loss of generality, we assume each gadget reduced from a hyperedge has weight $c_e=1$. Otherwise one can simply add $c_e$ as the coefficients and nothing needs to change. We omit the degree regularization term in \eqref{hypergraph-2-norm-cut}. Furthermore, we normalize the seeds to guarantee $\sum_{v\in V}d_v x_v=1$ and group terms in \eqref{hypergraph-2-norm-cut} to remove the source $x_s = 1$ and sink $x_t = 0$   
\begin{align}
    \minimize_{\vx} \quad Q(\vx)\triangleq \gamma \sum_{v\in V}d_v\left(x_v - \frac{1}{\vol(S)}1_S\right)^2 + \sum_{e\in \mathcal{E}}Q_e(\vx) \label{eq:obj-proof}
\end{align}
where 
\begin{align}\label{eq:qe}
\footnotesize
Q_e(\vx) \triangleq \underset{x_a^{(e)},x_b^{(e)}}{\min}\sum_{v\in e}\left[(x_v-x_a^{(e)})_+^2 +(x_b^{(e)}-x_v)_+^2\right] + \delta_e (x_a^{(e)}-x_b^{(e)})_+^2. 
\end{align}

We denote $M=\vol(\mathcal{H})=\sum_{v\in V} d_v$. We denote the solution $\vx$ with the parameter $\gamma$ and the seed set $S$ of the optimization \eqref{eq:obj-proof} as $\vx(\gamma, S)$ and its component for node $v$ as $x_v(\gamma, S)$. We also define another degree weighted vector $p = D\vx$, where $D$ is the diagonal degree matrix. For a vector $p$ and a node set $S'$, we use $p(S')$ to denote $p(S')=\sum_{v\in S'}p_v$. It is easy to check that $p(\gamma, S)(V) = 1$ for any $S$. For a node set $S$, define $\partial S=\{e\in \mathcal{E}|e\cap S\neq\emptyset, e\cap V\backslash S\neq\emptyset\}$. 

We now define our main tool: the Lov\'{a}sz-Simonovits Curve. 
\begin{definition}[Lov\'{a}sz-Simonovits Curve (LSC)]  Given an $\vx$, we order its components from large to small by breaking equality arbitrarily, say $x_1,x_2,....,x_{n}$ w.l.o.g, and define $S_j^{\vx}=\{x_1,...,x_j\}$. LSC defines a corresponding piece-wise function $I_{\vx}:[0,M] \rightarrow R $ s.t $I_{\vx}(0) = 0$, $I_{\vx}(\vol(G)) = 1$ and $I(\vol(S_j^{\vx})) = p(S_j^{\vx}).$ And for any $k\in [\vol(S_j^{\vx}),\vol(S_{j+1}^{\vx})]$, 
$$I_{\vx}(k) = I_{\vx}(\vol(S_j^{\vx}))+ (k-\vol(S_j^{\vx}))x_j.$$
\end{definition} 
Our proof depends on the following two Lemma~\ref{lm:up} and \ref{lm:lw} that will characterize the upper and lower bound of $I_{\vx}$, which finally leads to the main theorem.

\begin{lemma} \label{lm:up}
For a set $S$, given an $\vx=\vx(\gamma, S)$, let $\phi_{\vx}$ and $S_{\vx}$ be the minimal conductance and the node set obtained through a sweep-cut over $\vx$. For any integer $t>0$ and $k\in [0,M]$, the following bound holds 
$$I_{\vx}(k)\leq \frac{k}{M}+\frac{\gamma t}{2+\gamma}+\sqrt{\frac{\min(k,M-k)}{\min_{i\in S}d_i}}(1-\frac{\sigma_{\vx}^2\phi_{\vx}^2}{8})^t$$
where $\sigma_{\vx}=(2\max_{e\in \partial S_{\vx}}\min\{\delta_e, |e|/2\} + 1)^{-1}$.
\end{lemma} 

\begin{lemma} \label{lm:lw}
For a set $S$, if a node $v\in S$ is sampled according to a distribution $\mathbb{P}$ s.t 
\begin{align}\label{eq:sampling}
 \mathbb{E}_{v\sim \mathbb{P}} [p(\gamma,\{v\})(\bar{S})] \leq c p(\gamma, S)(\bar{S}),
 \end{align}
where $c$ is a constant, then with probability at least $\frac{1}{2},$ one has 
$$p(\gamma,\{v\})(S)\geq 1 - {2c\phi(S)}/{\gamma}.$$
\end{lemma} 

Lemma~\ref{lm:lw} gives the lower bound $I_{\vx(\gamma, \{v\})}(\vol(S))$ as  this value is no less than $p(\gamma,\{v\})(S)$. Note that the sampling assumption of Lemma~\ref{lm:lw} is natural in the standard graph case, when $\mathbb{P}$ samples each node proportionally to its degree, we have an equality with $c=1$ in \eqref{eq:sampling}. Combining Lemma~\ref{lm:up} and \ref{lm:lw},  we arrive at 
 \begin{theorem} \label{thm:app}
Given a set $S^*$ of vertices s.t. $\vol(S^*)\leq \frac{M}{2}$ and $\phi(S^*)\leq \frac{\gamma}{8c}$ for some positive constants $\gamma,\,c.$ If there exists a distribution $\mathbb{P}$ s.t. $\mathbb{E}_{v\sim \mathbb{P}}[p(\gamma,\{v\})](\bar{S}^*)\leq c p(\gamma,S^*)(\bar{S}^*),$ then with probability at least $\frac{1}{2},$ the obtained conductance satisfies
\begin{align*}
\phi_{\vx}\leq \sqrt{32\gamma \max_{e\in \partial S_{\vx}}\min\left\{\delta_e, \frac{|e|}{2}\right\} ln \left(100\frac{\vol(S^*)}{d_v}\right)},
\end{align*}
where $\vx=\vx(\gamma, \{v\})$ and $v$ is sampled from $\mathbb{P}$. $S_{\vx}$ is the node set obtained via the sweep-cut over $\vx$.
\end{theorem} 
\begin{proof} 
We combine Lemma~\ref{lm:lw} and \ref{lm:up} and use the same technique as Thm.~17 in \cite{panli2020qdsfm} (Sec. 7.7.3).
\end{proof} 

By removing the normalizing on the number of seeded nodes, Thm.~\ref{thm:app} becomes Thm.~\ref{thm:cond}.

\subsection{Proof of Lemma~\ref{lm:up}}
Define $L_e(\vx)\triangleq \nabla_{\vx} \frac{1}{2}Q_e(\vx)$ \eqref{eq:qe} and with some algebra, we have
\begin{align*}
    L_e(\vx) = \sum_{v\in e}\left[(x_v-x_a^{(e)*})_+ - (x_b^{(e)*}-x_v)_+\right]1_v,
\end{align*}
where $x_a^{(e)*}$ and $x_b^{(e)*}$ are the optimal values in~\eqref{eq:qe}. In the following, we will first prove Lemma~\ref{lm:diffusion} and further use it to prove Lemma~\ref{lm:diffusion2}. The same proof in Thm. 16 in \cite{panli2020qdsfm} (Sec. 7.7.2) can be used to leverage Lemma~\ref{lm:diffusion2} to prove Lemma~\ref{lm:up}.

\begin{lemma}\label{lm:diffusion} 
Given an $\vx$, we order its components from large to small, say $x_1,x_2,....,x_{n}$ w.l.o.g., and define $S_j^{\vx}=\{x_1,...,x_j\}$. Define $\sigma_{j}^{\vx} = (1+2\max_{e\in \partial S_j^{\vx}}\min\left\{\delta_e, |e|/2\right\})^{-1},$ and we have
\begin{align*}
    2I_{\vx}(\vol(S_j^{\vx})) - \langle \sum_{e\in E} L_e(x), 1_{S_j^{\vx}}\rangle &\leq \\ I_{\vx}(\vol(S_j^{\vx})-\sigma_{j}^{\vx} cut(S_j^{\vx}))&+I_{\vx}(\vol(S_j^{\vx})+\sigma_{j}^{\vx} cut(S_j^{\vx})).
\end{align*}
\end{lemma} 

\begin{lemma} \label{lm:diffusion2} 
Suppose $\vx=\vx(\gamma, S)$, $\vx_0 = 1_S/\vol(S)$ and $\sigma_{j}^{\vx} = (1+2\max_{e\in \partial S_j^{\vx}}\min\{\delta_e,|e|/2\})^{-1}$. We have
\begin{multline}\label{eq:diffution2} 
I_{\vx}(\vol(S_j^{\vx}))\leq \frac{\gamma}{2+\gamma}I_{\vx_0}(S_j^{\vx_0})+\frac{1}{2+\gamma}(I_{\vx}(\vol(S_j^{\vx})-\sigma_j^{\vx} cut(S_j^{\vx}))+\\I_{\vx}(\vol(S_j^{\vx})+\sigma_j^{\vx} cut(S_j^{\vx}))).
\end{multline} 
Furthermore, for $k\in [0,M],$ $I_{\vx}(k) \leq I_{\vx_0}(k).$
\end{lemma} 

\mypara{Proof of Lemma~\ref{lm:diffusion}.}
Given a hyperedge $e$, we order $\{x_v \mid v\in e\}$ from large to small by breaking equality arbitrarily and obtain $x_1^{(e)},x_2^{(e)},...,x_e^{(e)}$. Suppose $x_k^{(e)}\in S_j^{\vx}$ and $x_{k+1}^{(e)}\notin S_j^{\vx}$. Then, 
\begin{align}\label{eq:one-e-diff}
    \textstyle \langle L_e(x), 1_{S_j^{\vx}}\rangle = \sum\limits_{i=1}^k \bigl[(x_i^{(e)}-x_a^{(e)*})_+ - (x_b^{(e)*}-x_i^{(e)})_+\bigr]
\end{align}
Next, we will bound \eqref{eq:one-e-diff} by analyzing three cases. We only focus on $x_a^{(e)*} > x_b^{(e)*}$ and otherwise $x_v$'s are constant for all $v\in e$. Also denote $k_+=\max\{i \mid x_i^{(e)}>x_a^{(e)*}\}$, $k_- = \min\{e+1-i \mid x_i^{(e)}<x_b^{(e)*}\}$, and $k^-=|e|+1-k_-$. By the optimality of $x_a^{(e)*},\,x_b^{(e)*}$, we have 
\[
        x_a^{(e)*} = (k_- + \delta_e)X_1^{k_+} + \delta_e X_{k^-}^{|e|},\quad
    x_b^{(e)*} = \delta_eX_1^{k_+} + (k_++\delta_e)X_{k^-}^{|e|}, 
\] 
where $X_{i_1}^{i_2}= (k_+k_-+\delta_e(k_++k_-))^{-1}\textstyle \sum_{i=i_1}^{i_2} x_i^{(e)}.$

Thus, we have 
\begin{equation}\label{eq:3-case}
\langle L_e(x), 1_{S_j^{\vx}}\rangle \! = \!
\footnotesize
    \left\{\begin{array}{@{}l@{\;\;}l@{}}[(k_+-k)(k_-+\delta_e)+k_-\delta_e]X_{1}^k & \text{if}\,k\leq k_+ \\
    \quad - k(k_-+\delta_e)X_{k+1}^{k_+} - k\delta_e X_{k^-}^{|e|}, &  \\
    k_-\delta_eX_{1}^{k_+} - k_+\delta_eX_{k^-}^{|e|}, & \text{if}\, k_+\leq k\leq k^- \\ 
    (|e|-k)\delta_e X_{1}^{k_+} + (|e|-k)(k_++\delta_e) X_{k^-}^{k} & \text{if}\,k\geq k^-.\\
    \quad - [k_+\delta_e + (k-k^-)(k_++\delta_e)] X_{k}^{|e|} 
    \end{array}\right.    
\end{equation}
By using the definition of $X_{i_1}^{i_2}$, noticing that all coefficients on $x_{i}^{(e)} \le 1$ in the left hand side of~\eqref{eq:3-case},   and a good deal of algebra, we can further show
\begin{equation*}
\small
    \left\{\begin{array}{@{}ll@{}}\frac{k(k_+-k)(k_-+\delta_e)+k_-\delta_e}{k_+k_-+\delta_e(k_++k_-)}\geq \frac{2}{\delta_e' + 2}\min\{k,|e|-k,\delta_e\} & \text{if}\,k\leq k_+ \\
    \frac{k_+k_-\delta_e}{k_+k_-+\delta_e(k_++k_-)}\geq \frac{1}{2\delta_e' + 1}\min\{k,|e|-k,\delta_e\} , & \text{if}\, k_+\leq k\leq k^- \\
    \frac{k_-[k_+\delta_e + (k-k^-)(k_++\delta_e)]}{k_+k_-+\delta_e(k_++k_-)} \geq \frac{2}{\delta_e' + 2}\min\{k,|e|-k,\delta_e\}, & \text{if}\,k\geq k^-. 
    \end{array}\right.
\end{equation*}
where $\delta_e'=\min\{\delta_e,|e|/2\}$. In each case of \eqref{eq:3-case}, the sum of positive coefficient before each $x_{i}^{(e)}$ equals to the sum of negative coefficients, which are both lower bounded by $\frac{1}{2\delta_e' + 1}$ times the splitting cost of $e$.
Therefore,
\begin{align*}
    2I_{\vx}(\vol(S_j^{\vx})) - \sum_{e\in E}\langle L_e(x), 1_{S_j^{\vx}}\rangle &\leq \\
    I_{\vx}(\vol(S_j^{\vx})-\sigma_{j}^{\vx} cut(S_j^{\vx}))&+I_{\vx}(\vol(S_j^{\vx})+\sigma_{j}^{\vx}cut(S_j^{\vx})),
\end{align*}
where $\sigma_{j}^{\vx} = 1/(2\max_{e\in \partial S_j^{\vx}}\delta_e' + 1)$, which concludes the proof. 

\mypara{Proof of Lemma~\ref{lm:diffusion2}.}
Compute the derivative of $Q(\vx)$ \eqref{eq:obj-proof} w.r.t. $\vx$ and use the optimality of $\vx=\vx(\gamma,S)$ to get 
$0=  \gamma D(\vx -\vx_0) + L_e(\vx)$. Therefore,the  inner product with $1_{S_j^{\vx}}$
\begin{align*}
    0 = \gamma (I_{\vx}(\vol(S_j^{\vx}))- I_{\vx_0}(\vol(S_j^{\vx_0}))) + \langle L_e(\vx), 1_{S_j^{\vx}}\rangle.
\end{align*}
Plug in Lemma~\ref{lm:diffusion} and we achieve \eqref{eq:diffution2}. By using the concavity of $I_{\vx}$, we obtain $I_{\vx}(\vol(S_j^{\vx}))\leq  I_{\vx_0}(\vol(S_j^{\vx_0}))$ and therefore $I_{\vx}(k)\leq  I_{\vx_0}(k)$ for any $k\in[0,M]$.

\subsection{Proof of Lemma~\ref{lm:lw}}
If the following Lemma~\ref{lm:lw-piS} is true, then we have $\mathbb{E}_{v\sim \mathbb{P}} [p(\gamma,\{v\})(\bar{S})] \leq c p(\gamma,S)(\bar{S}) = c(1 -p(\gamma, S)(S))\leq  {c\phi(S)}/{\gamma}$. By using Markov's inequality, with probability $\frac{1}{2}$, we will sample a node $v$ such that $p(\gamma,\{v\})(\bar{S})\leq {2c\phi(S)}/{\gamma}$, which concludes the proof. 
\begin{lemma}\label{lm:lw-piS}
For any $ S\subset V$, $p(\gamma, S)(S)\geq 1-{\phi(S)}/{\gamma}$.
\end{lemma}

\mypara{Proof.} 
This mass from the nodes in $S$ to the nodes in $\bar{S}$ naturally diffuses from the auxiliary nodes $v_a^{(e)}$ to $v_b^{(e)}$ for $e\in \partial S$. As need to lower bound $p(\gamma, S)(S)$, we may consider fixing $x_v=0, \forall v\in \bar{S}$ of $Q(\vx)$ and the obtained solution $\tilde{\vx}$ naturally satisfies 
\begin{align*}
    p(\gamma, S)(S) \geq \sum\nolimits_{v\in S} d_v\tilde{x}_v, \;\text{where}\;
    \tilde{\vx} \triangleq \arg\min_{\vx} Q(\vx)|_{x_v=0, \forall v\in \bar{S}}.
\end{align*}
The optimality of $\tilde{x}_{a}^{(e)}, \tilde{x}_{b}^{(e)}$ for $e\in\partial S$ in this case implies
\begin{align*}
 \textstyle -\sum_{v\in e} (\tilde{x}_v - \tilde{x}_{a}^{(e)})_+ + \delta_e (\tilde{x}_{a}^{(e)}-\tilde{x}_{b}^{(e)})_+ = 0, \\
 \textstyle \sum_{v\in e} (-\tilde{x}_v + \tilde{x}_{b}^{(e)})_+ - \delta_e (\tilde{x}_{a}^{(e)}-\tilde{x}_{b}^{(e)})_+ = 0
\end{align*}
As $\tilde{x}_{v}=0$ for $v\in e\backslash S$ and $\tilde{x}_{v}\leq \frac{1}{\vol(S)}$ for $v\in e \cap S$, we have 
\begin{align*} 
\tilde{x}_{b}^{(e)} \geq  \frac{\delta_e}{\delta_e + |e\backslash S|}  \tilde{x}_{a}^{(e)},\;
    \tilde{x}_{a}^{(e)} &\leq \frac{|e\cap S|}{\delta_e + |e\cap S|}\frac{1}{\vol(S)},
\end{align*}
and further, $\forall e\in \partial S$,
\begin{align} \label{eq:xtildea}
    &\textstyle \sum_{v\in e\cap S} (\tilde{x}_v - \tilde{x}_{a}^{(e)})_+ = \delta_e (\tilde{x}_{a}^{(e)} - \tilde{x}_{b}^{(e)})_+  \\
    &\leq \frac{|e\cap S||e\backslash S|\delta_e}{(\delta_e + |e\cap S|)(\delta_e + |e\backslash S|)}\frac{1}{\vol(S)} \leq \frac{\min\{|e\cap S|,|e\backslash S|,\delta_e\}}{\vol(S)}. \nonumber
\end{align}
The optimality of $\tilde{\vx}$ implies
\begin{align} \label{eq:xtilde}
\sum_{v\in S} \gamma d_v \left(\tilde{x}_v - \frac{1}{\vol(S)}\right) + \sum_{e\in \partial S} \sum_{v\in e\cap S} (\tilde{x}_v - \tilde{x}_{a}^{(e)})_+ =0.
\end{align}
Here we use that $\sum_{e \subset S}\sum_{v\in e} [(\tilde{x}_v - \tilde{x}_{a}^{(e)})_+ + ( \tilde{x}_{b}^{(e)}- \tilde{x}_v)_+]=0$. Plug \eqref{eq:xtildea} into \eqref{eq:xtilde} and we have 
\begin{align*}
0 \! \leq \! \gamma \left(\sum_{v\in S} d_v \tilde{x}_v - 1\right) + \cancelto{\phi(S)}{\sum_{e\in \partial S} \frac{\min\{|e\cap S|,|e\backslash S|,\delta_e\}}{\vol(S)}} ,
\end{align*}
which concludes the proof.

\confspace{-0.5\baselineskip}
\section{{Selecting $\delta$.}}
\label{sec:select-delta}


To select $\delta$ for each dataset, we run LH-2.0 on a handful of alternative clusters as we vary $\delta$. Below, we show F1 scores on those clusters and pick 
$\delta=1$ for Amazon and $\delta=1000$ for Stack Overflow. 

\begin{minipage}{0.4\linewidth}%
\begin{picture}(100,55)
\put(0,0){\includegraphics[width=\linewidth]{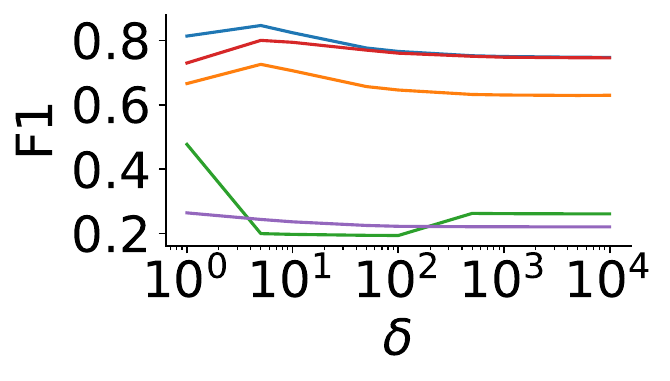}}
\put(50,50){\footnotesize Amazon}
\end{picture}
\end{minipage}%
\hspace{2em}
\begin{minipage}{0.4\linewidth}%
\begin{picture}(100,55)
\put(0,0){\includegraphics[width=\linewidth]{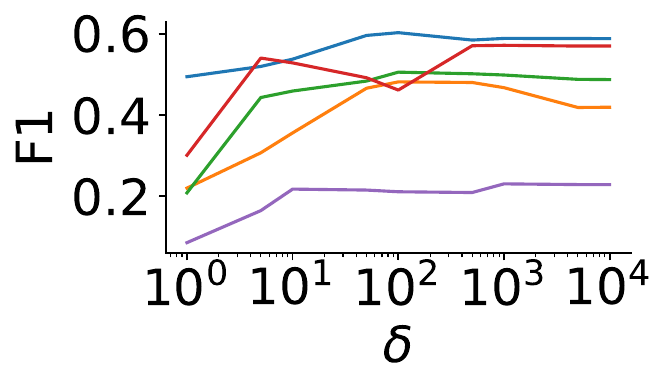}}
\put(30,50){\footnotesize Stack Overflow}
\end{picture}
\end{minipage}%
\end{document}